\documentclass[11pt]{amsart}

\usepackage{amsmath, amssymb, amsthm, amscd}
\usepackage{url}
\usepackage{mathptmx}
\usepackage{graphicx}
\usepackage[hidelinks,pagebackref,pdftex]{hyperref}
\usepackage{color}
\usepackage{import}
\usepackage[margin=1in]{geometry}
\usepackage{caption}
\usepackage{bbold}
\usepackage{cancel}
\usepackage{tikz}

\usetikzlibrary{decorations.pathmorphing,decorations.markings}
\usetikzlibrary{decorations.markings}

\usepackage[linesnumbered,lined,ruled]{algorithm2e}

\SetCommentSty{mycommfont}

 \renewcommand*{\backref}[1]{}
 \renewcommand*{\backrefalt}[4]{
   \ifcase #1
   [No citations.]
   \or [#2]
   \else [#2]
   \fi }

 \AtBeginDocument{%
    \def\MR#1{}
 }

\numberwithin{equation}{section}
\theoremstyle{plain}
\newtheorem{theorem}[equation]{Theorem}

\newtheorem{lemma}[equation]{Lemma}
\newtheorem{corollary}[equation]{Corollary}

\newtheorem*{result}{Main Result}

\newtheorem*{namedtheorem}{\theoremname}
\newcommand{\theoremname}{testing}

\theoremstyle{definition}

\renewcommand{\paragraph}[1]{\medskip\noindent {\bf #1.}}

\newcommand{\sixj}[6]{
  \left|\begin{array}{ccccc}
  #1 & #2 & #3\\
  #4 & #5 & #6\\
  \end{array}\right|
}

\newcommand{\qcap}[2]{
\begin{scope}[decoration={
    markings,
    mark=at position 0.75 with {\arrow{>}}}
    ]
    \draw[line width=1,postaction={decorate}] (#1,#2) arc (180:0:1);
    \end{scope}}
\newcommand{\qcup}[2]{
\begin{scope}[decoration={
    markings,
    mark=at position 0.25 with {\arrow{>}}}
    ]
\draw[line width=1,postaction={decorate}] (#1,#2) arc (180:360:1);
\end{scope}}
\newcommand{\tetsymbleft}[6]{
\qcap{0}{0}
\qcup{0}{-2}
\begin{scope}[decoration={
    markings,
    mark=at position 0.5 with {\arrow{>}}}
    ]
\draw[line width=1,postaction={decorate}] (0,0) node[left] {$#1$} --(0,-2) node[left,below, xshift=-4, yshift=-4] {$#3$};
\draw[line width=1,postaction={decorate}] (2,0) node[right,above,xshift=3,yshift=2] {$#6$} -- (2,-2) node[right] {$#4$};
\draw (3.5,-1) node {$#5$};
\end{scope}
\begin{scope}[decoration={
    markings,
    mark=at position 0.5 with {\arrow{<}}}
    ]
    \draw[line width=1,postaction={decorate}] (0,-2)--(2,0) node[midway,above] {$#2$};
\draw[line width=1,postaction={decorate}] (1,1) to  [out=90, in=90] (3,1)--(3,-3) to [out=-90, in=-90] (1,-3);
\end{scope}}

\newcommand{\tetsymbright}[6]{
\begin{scope}[xscale=-1,xshift=-2cm]
\qcap{0}{0}
\end{scope}
\begin{scope}[xscale=-1,xshift=-2cm]
\qcup{0}{-2}
\end{scope}
\begin{scope}[decoration={
    markings,
    mark=at position 0.5 with {\arrow{>}}}
    ]
\draw[line width=1,postaction={decorate}] (0,0) --(0,-2) node[left,below, xshift=-4, yshift=-4] {$#1$};
\draw[line width=1,postaction={decorate}] (2,0) node[right,above,xshift=3,yshift=2] {$#4$} -- (2,-2) node[right] {$#6$};
\draw (0.25,1) node[left] {$#3$};
\draw (3.5,-1) node {$#5$};
\end{scope}
\begin{scope}[decoration={
    markings,
    mark=at position 0.5 with {\arrow{<}}}
    ]
    \draw[line width=1,postaction={decorate}] (2,-2)--(0,0) node[midway,above] {$#2$};
\draw[line width=1,postaction={decorate}] (1,1) to  [out=90, in=90] (3,1)--(3,-3) to [out=-90, in=-90] (1,-3);
\end{scope}
}

\newcommand{\Vect}{\mathrm{Vec}}
\newcommand{\Chain}{C}

\newcommand{\Cycle}{Z}

\newcommand{\M}{M}
\newcommand{\tri}{\mathfrak{T}}
\newcommand{\Z}{\mathbb{Z}}
\newcommand{\R}{\mathbb{R}}
\newcommand{\Q}{\mathbb{Q}}
\newcommand{\C}{\mathbb{C}}
\newcommand{\ty}{\operatorname{TY}}
\newcommand{\poly}{\operatorname{poly}}
\newcommand{\adm}{\operatorname{Adm}}

\newcommand{\calP}{\mathcal{P}}
\newcommand{\calQ}{\mathcal{Q}}

\newcommand{\calC}{\mathcal{C}}
\newcommand{\calB}{\mathcal{B}}

\newcommand{\e}{\mathbf{e}}

\newcommand{\tv}{\operatorname{TV}}

\newcommand{\Ztwo}{\Z/2\Z}
\newcommand{\rad}[2]{\operatorname{rad}(#1,#2)}

\newcommand{\row}{\operatorname{row}}
\newcommand{\col}{\operatorname{col}}

\newcommand{\Eemp}{E_{\emptyset}}
\newcommand{\Eone}{E_{{\bullet}}}
\newcommand{\Femp}{F_{\emptyset}}
\newcommand{\Ftwo}{F_{\bullet\bullet}}
\newcommand{\Temp}{T_{\emptyset}}

\font\Bigmath=cmsy10 scaled \magstep2
\font\bigmath=cmsy10 scaled \magstep1
\def\diamondplusone{\mathrel{%
  \ooalign{\raise.25ex\hbox{\hspace{0.015cm}$\scriptscriptstyle+$}\cr\hss$\diamond$\hss}}}
\def\diamondplustwo{\mathrel{%
  \ooalign{$+$\cr\hss\lower.255ex\hbox{\Bigmath\char5}\hss}}}
\def\diamondplusthree{\mathrel{%
  \ooalign{$\scriptstyle+$\cr\hss\lower.29ex\hbox{\bigmath\char5}\hss}}}
\def\diamondminusone{\mathrel{%
  \ooalign{\raise.25ex\hbox{\hspace{0.015cm}$\scriptscriptstyle-$}\cr\hss$\diamond$\hss}}}
\def\diamondminustwo{\mathrel{%
  \ooalign{$-$\cr\hss\lower.255ex\hbox{\Bigmath\char5}\hss}}}
\def\diamondminusthree{\mathrel{%
  \ooalign{$\scriptstyle-$\cr\hss\lower.29ex\hbox{\bigmath\char5}\hss}}}
\newcommand{\diamondplus}{\diamondplusone}
\newcommand{\diamondminus}{\diamondminusone}

\newcommand{\Ttriang}{T_{\triangle}}
\newcommand{\Tquadplus}{T_{\boxplus}}
\newcommand{\Tquadminus}{T_{\boxminus}}
\newcommand{\Tcrookquadplus}{T_{\diamondplus}}
\newcommand{\Tcrookquadminus}{T_{\diamondminus}}

\newcommand{\ifff}{{\em iff} }
\newcommand{\ie}{{\em i.e.}}
\newcommand{\eg}{{\em e.g.}}

\definecolor{ericgreen}{rgb}{0.1, 0.75, .1}

\begin{document}

\title[FPT algorithm for Tambara-Yamagami invariants]{An algorithm for Tambara-Yamagami quantum invariants of 3-manifolds, parameterized by the first Betti number}

\author{Colleen Delaney, Cl\'ement Maria and Eric Samperton}

\address{Colleen Delaney: Purdue University, Departments of Mathematics and Physics}
\email{colleend@purdue.du}

\address{Cl\'ement Maria: INRIA Universit\'e C\^ote d'Azur \& FGV/EMAp Rio de Janeiro}
\email{clement.maria@inria.fr}

\address{Eric Samperton: Purdue University, Departments of Mathematics and Computer Science}
\email{eric@purdue.edu}

\begin{abstract}
Quantum topology provides various frameworks for defining and computing invariants of manifolds inspired by quantum theory. One such framework of substantial interest in both mathematics and physics is the Turaev-Viro-Barrett-Westbury state sum construction, which uses the data of a spherical fusion category to define topological invariants of triangulated 3-manifolds via tensor network contractions. In this work we analyze the computational complexity of state sum invariants of 3-manifolds derived from Tambara-Yamagami categories. While these categories are the simplest source of state sum invariants beyond finite abelian groups (whose invariants can be computed in polynomial time) their computational complexities are yet to be fully understood. We first establish that the invariants arising from even the smallest Tambara-Yamagami categories are \#P-hard to compute, so that one expects the same to be true of the whole family. Our main result is then the existence of a fixed parameter tractable algorithm to compute these 3-manifold invariants, where the parameter is the first Betti number of the 3-manifold with $\mathbb{Z}/2\mathbb{Z}$ coefficients. 

Contrary to other domains of computational topology, such as graphs on surfaces, very few hard problems in 3-manifold topology are known to admit FPT algorithms with a topological parameter. However, such algorithms are of particular interest as their complexity depends only polynomially on the combinatorial representation of the input, regardless of size or combinatorial width. Additionally, in the case of Betti numbers, the parameter itself is computable in polynomial time. Thus while one generally expects quantum invariants to be hard to compute classically, our results suggest that the hardness of computing state sum invariants from Tambara-Yamagami categories arises from classical 3-manifold topology rather than the quantum nature of the algebraic input.
\end{abstract}

\maketitle

 \tableofcontents{}

\section{Introduction}
\label{sec:introduction}

Quantum physics---especially quantum field theory---is a profuse source of invariants in low-dimensional topology, and these invariants often provide information about manifolds beyond the reach of classical topology.  However, the native paradigm for (approximately) computing them is manifestly quantum computational, so that in general one should expect them to be hard to compute on a classical computer \cite{Kuperberg:hard}. Naturally then we seek to understand the full landscape of the computational complexities of quantum invariants. From the point of view of computational topology we are particularly interested in a middle ground where there may be interesting invariants that are not too hard to compute, and to develop algorithms to compute them as efficiently as possible. 

In this work, we investigate the complexity landscape of a specific class of quantum invariants of 3-manifolds, namely, the \emph{Turaev-Viro-Barrett-Westbury (TVBW) state sum construction} \cite{TuraevViro,BarrettWestbury}.  Unlike other quantum invariants (e.g.~Donaldson invariants), the TVBW construction is \emph{a priori} algorithmically computable, combining a finite amount of algebraic data in the form of a {\em (skeletalized) spherical fusion category} $\calC$ with a triangulation $\tri$ of a closed, oriented 3-manifold $\M$ to generate a complex algebraic number $|\M|_\calC \in \C$ that is an orientation-preserving homeomorphism invariant of $\M$.  The invariant $|\M|_\calC$ is defined as the contraction of a certain tensor network built by decorating $\tri$ with the data of $\calC$, and can be formally described via a combinatorial {\em state sum formula}.  We refer the reader to Appendix \ref{sec:quantumtopology} for a brief review of the TVBW construction (in the case of multiplicity-free, pseudo-unitary spherical categories) and \cite{Turaev:book} for more background.

As briefly reviewed at the end of this introduction, it is desirable for various reasons to understand the computational complexity of evaluating the TVBW invariant $|\cdot|_\calC$ on a given triangulated 3-manifold
\[
\begin{aligned}
|\cdot|_\calC: \{\text{triangulations of closed, oriented 3-manifolds}\} &\to \mathbb{C} \\
\M &\mapsto |\M|_\calC.
\end{aligned}
\]
More precisely, we would like to understand how the computation of these invariants depends on the choice of spherical fusion category $\calC$, since there are countably infinitely many equivalence classes of such categories, and their classification is at least as hard as the classification of finite groups.  As usual, we aim to understand the complexity of computing $|M|_\calC$ as a function of the size of a triangulation $\tri$ used to encode $M$. 

We approach this complexity-theoretic classification problem from the perspective that fusion categories are ``quantum'' generalizations of (the representation theory of) finite groups.  While this perspective has not been made explicitly overt in the computational topology literature, it is common in quantum algebra, where the extent to which a fusion category deviates from a finite group $G$ can be made precise in various ways.  For the sake of space, we will not review these notions here, but will note that anyway one cuts it, the simplest spherical fusion categories are of the type $\Vect(A)$---where $A$ is an abelian group---consisting of $A$-graded finite-dimensional vector spaces, with a tensor product that comes from the group operation in $A$.  In fact, for any finite group $G$, the category of $G$-graded vector spaces $\Vect(G)$ forms a spherical fusion category, but if $G$ is non-abelian, we consider this category to be more ``complicated" than $\Vect(A)$.  One reason for this is that $|\cdot|_{\Vect(A)}$ is always polynomial-time computable, whereas $|\cdot|_{\Vect(G)}$ is often $\#P$-hard for non-abelian groups \cite{KuperbergSamperton2}.

Arguably, after $\Vect(A)$, the next simplest family of fusion categories consists of the Tambara-Yamagami categories $\ty(A,\chi,\nu)$.  Such a category is determined by a finite abelian group $A$, together with a small amount of additional data: a \emph{bicharacter} (\emph{i.e.}, a non-degenerate symmetric bilinear pairing) $\chi: A \times A \to U(1)$, and a choice of square root $\nu = \pm1/ \sqrt{|A|} = \pm |A|^{-1/2}$ \cite{TambaraY98}.  (See Appendix \ref{sec:quantumtopology} for a complete description of these categories using this data.) And yet, in Appendix \ref{sec:hard} we show that already in the simplest non-trivial case with $A=\Z/2\Z$, the invariants $|\cdot|_{\ty(\Ztwo,\chi,-1/\sqrt{2})}$ are $\#P$-hard to compute.  This is more-or-less a restatement of known results about the Turaev-Viro invariant $\tv_4$ \cite{KirbyM04,MariaS20}, since $|\cdot|_{\ty(\Ztwo,\chi,1/\sqrt{2})} = \tv_4(\cdot)$.  However, as we shall explain, our approach offers a conceptual clarification by placing $\tv_4$ within the broader taxonomy of spherical fusion categories.  In particular, we conjecture that the TVBW state sum invariant of \emph{every} non-trivial Tambara-Yamagami category $\ty(A,\chi,\nu)$ is $\#\mathsf{P}$-hard. 

Our main result is a parameterized algorithm to compute the TVBW state sum invariants $|\M|_{\ty(A,\chi,\nu)}$ associated to Tambara-Yamagami categories, where our parameter is the first $\Ztwo$ Betti number $\beta_1$. 

\begin{result} [Informal; see Theorem~\ref{thm:maincpx}.]
Fix an integer $B\ge 0$ and a Tambara-Yamagami category $\ty(A, \chi, \nu)$.  Then there exists a polynomial time algorithm to evaluate $|\M|_{\ty(A,\chi,\nu)}$ for triangulated 3-manifolds with $\beta_1(\M) \le B$. 
\end{result}

\noindent While the hardness of general Tambara-Yamagami invariants is still conjectural, we interpret our main result as showing that any such hardness is explainable by a simple classical topological fact: there exist triangulations $\tri$ of 3-manifolds $M$ with $n$ simplices where $\beta_1(\M)$ can be as large as $\Theta(n)$.

Our algorithm succeeds by converting the natively quantum topological computational problem of computing a TVBW state sum invariant of a triangulated 3-manifold ---which \emph{a priori} requires an expensive tensor network contraction--- into a sequence of other much more efficient classical computational problems in algebraic number theory and combinatorial-algebraic topology.  Prior to our work, to the best of our knowledge, the only published algorithm in 3-manifold topology directly parameterized by an efficiently computable \emph{topological} quantity (whose value is independent of the presentation of the input) is the algorithm of Maria and Spreer to compute the Turaev-Viro quantum invariant $\tv_4$ of 3-manifolds~\cite{TuraevViro} at a 4th root of unity in $O(2^{\beta_1} n^3)$ operations~\cite{MariaS20}.

The technical underpinnings of our algorithm are explained in Section~\ref{sec:fptalgorithm}; the algorithm itself is precisely described and analyzed in Section \ref{sec:complexity}.  At the heart of our algorithm is a (non-obvious!) observation: $|\M|_{\ty(A,\chi,\nu)}$ can be identified with a (normalized) sum of $|H^1(\M,\Z/2\Z)| = 2^{\beta_1}$ many \emph{Gauss sums of $\Q/\Z$-valued quadratic forms} on certain finite abelian groups derived from $A$ and the given triangulation of $M$.
\begin{align*}
	|\M|_{\ty(A,\chi,\nu)}
	&= \sum_{\theta \in \adm(\tri)} |\tri|_\theta && \substack{\text{(def.~of TVBW state} \\ \text{sum invariant; Sec.~\ref{sec:tyquantuminvariant})}} \\
	&= \sum_{\alpha \in Z^1(\tri, \Ztwo)} \underbrace{\sum_{\theta \in \adm(\tri,\alpha)} |\tri|_\theta}_{\substack{\text{\emph{partial state sum at}} \\ \text{\emph{cocycle $\alpha$}; Sec.~\ref{ss:partial}}}} && \substack{\text{($\Ztwo$-grading of} \\ \text{$\ty(A,\chi,\nu)$; Sec.~\ref{ss:cocyle})}} \\
	&=  \#B^1(\tri,\Z/2\Z)\cdot \sum_{[\alpha] \in H^1(\tri,\Z/2\Z)} \underbrace{\sum_{\theta \in \adm(\tri,\alpha)} |\tri|_\theta}_{\substack{(\star) \ \text{poly. time computable as a} \\ \text{quadratic Gauss sum}}} && \substack{\text{(follows from \cite{turaev20123};}\\ \text{Sec.~\ref{ss:cohomologous})}}  
\end{align*}
\noindent
The second and third equalities above follow for a rather general reason: $\ty(A,\chi,\nu)$ possesses the structure of a $\Ztwo$-{\em grading} and so by general facts related to {\it homotopy quantum field theory} \cite{turaev20123}, we are able to reduce $|M|_{\ty(A,\chi,\nu)}$ to a sum $2^{\beta_1}$ terms in a manner similar in spirit to \cite{MariaS20}.  Notably, however, the concept of \emph{graded fusion categories} is absent from \cite{MariaS20}, and this structure enables us to analyze the state sum invariant of $M$ directly from \emph{any} simplicial triangulation, without having to pass to a one-vertex triangulation via the crushing algorithm of Jaco and Rubenstein \cite{JacoR03}.

The bulk of Section \ref{sec:fptalgorithm} is then concerned with explaining and justifying $(\star)$, and our methods here diverge substantially from \cite{MariaS20}.  We remark that the efficient algorithmic evaluation of quadratic Gauss sums appears to be well-known to experts in computational algebraic number theory, although we were unable to find any specific references for this fact, the closest apparently being~\cite{BasakJ15}; see Section \ref{ss:nt} and Appendix \ref{sec:gausssum} for further discussion.

The general nature of our main result allows us to establish new conceptual insight into the complexity-theoretic landscape of TVBW invariants, a portion of which we summarize in Figure \ref{fig:cpxquantumtopo}.  It is known that TVBW 3-manifold invariants are often \#P-hard to compute exactly, and sometimes are hard even just to approximate~\cite{JaegerVertiganWelsh, KirbyM04, KroviRussell, Kuperberg:hard, KuperbergSamperton1, KuperbergSamperton2}; moreover, in many of these cases, it is known that there can be no tractable algorithm in $\beta_1$.  Thus the existence of a parametrized algorithm in a topological and efficiently computable parameter for the Tambara-Yamagami invariants is an interesting counterpoint to their expected hardness. This suggests that the hardness of computing the invariants is due to some intrinsic hardness in 3-manifold topology, rather than the Tambara-Yamagami category itself, whose construction based on a finite abelian group is quite simple compared to more general spherical fusion categories.

\begin{figure}[t]
	\begin{tabular}{c|c c c}
		spherical category $\mathcal{C}$ & complexity of $|M|_\mathcal{C}$ & FPT in $\beta_1$? & Reference \\
		\hline
		$\Vect(A)$ for $A$ abelian & in $\mathsf{fP}$ & Yes & well-known, e.g.~\cite{KuperbergSamperton1} \\
		$\ty(\Z/2\Z,\exp(\pi i a b), -1/\sqrt{2})$ & $\#\mathsf{P}$-hard & Yes & \cite{KirbyM04,MariaS20}, App.~\ref{sec:hard} \\
		$\ty(A,\chi,\nu)$ & unknown in general & Yes & this work \\
		$\Vect(G)$ for $G$ non-abelian solvable & unknown & unknown & \\
		$\Vect(G)$ for $G$ non-abelian simple & $\#\mathsf{P}$-complete & No$^*$& \cite{KuperbergSamperton1} \\
		$U_q\mathfrak{sl}_2$-mod ($q$ a non-lattice root of 1) & $\#\mathsf{P}$-hard & No$^*$  & \cite{JaegerVertiganWelsh,Kuperberg:hard}
	\end{tabular}
	\caption{State of the art for  computing TVBW invariants based on ``simple'' spherical fusion categories.  $^*$(unless $\mathsf{fP} = \#\mathsf{P}$)}
	\label{fig:cpxquantumtopo}
\end{figure}

After scrutinizing Figure \ref{fig:cpxquantumtopo}, it seems sensible to expect that there should be a dichotomy theorem for TVBW invariants.  A na\"ive conjecture would be that $|\cdot|_{\calC}$ is in $\mathsf{fP}$ if and only if $\calC\simeq\Vect(A)$, and otherwise $|\cdot|_{\calC}$ is $\#\mathsf{P}$-hard.  However, this cannot be correct, as there are ``twisted" versions of $A$-graded vector spaces $\Vect(A,\eta)$, $\eta \in H^3(A,U(1))$, for which $|\cdot|_{\Vect(A,\eta)}$ is also known to be polynomial-time computable.  More subtly, if $G$ is any 2-step nilpotent group, then $|M|_{\Vect(G)} = \#\{\pi_1(M) \to G\}/\#G$ should also be in $\mathsf{fP}$.

A systematic understanding of the computational complexity of state sum invariants is significant because they fit into the larger framework of a fully-extended 3D TQFT. Indeed, this connection endows the invariants with important applications to {\em fault-tolerant quantum computation} and the characterization of {\em topologically ordered phases of matter}, e.g.~via {\em topological quantum computing}~\cite{Kitaev:anyons,KonigKuperbergReichardt:codes,RowellWang:bulletin,WangCBMS}, as well as to the developing paradigm of  ``non-invertible'' or ``categorical'' symmetries and dualities present in {\em quantum field theories} and {\em lattice models}~\cite{AasenFendleyMong,StandardModel,FreedMooreTeleman,FreedTeleman,ThorngrenWang}. Of course, there is plenty of motivation from low-dimensional topology as well, where these invariants can be used in practice to distinguish between non-homeomorphic 3-manifolds.  TVBW invariants are heuristically strong, with a typical category $\calC$ apparently capable of distinguishing ``most'' pairs of non-homeomorphic manifolds.\footnote{That said, every $\calC$ does in fact admit infinitely many pairs of non-homeomorphic manifolds that it cannot distinguish; moreover, Funar has shown the existence of pairs of distinct manifolds $\M_1 \ne \M_2$ such that $|M_1|_\calC = |M_2|_\calC$ for all spherical fusion categories $\calC$~\cite{Funar}.} Notably, the Tambara-Yamagami invariants studied in this article are capable of distinguishing homotopy equivalent lens spaces that are not homeomorphic, cf.~\cite[Thm.~2]{MariaS20}.

\section{Background}
\label{sec:background}
In this section we introduce the TVBW invariants of 3-manifolds derived from Tambara-Yamagami categories, as well as some of the basic components necessary for our algorithm to compute them.

\subsection{Elements of algorithmic number theory} 
\label{ss:nt}

\paragraph{Finite abelian groups} Let $(A,+)$ be a {\em finite abelian group} with $|A|$ elements, denoted additively. The fundamental theorem of finitely generated abelian groups shows that $A$ admits a unique {\it primary decomposition} of the form:
\[
  A \cong \bigoplus_{\text{$p$ prime}} \ \ \bigoplus_{k_i,n_i\,:\, i \in I_p} \left(\Z/p^{k_i}\Z\right)^{n_i}
\] 
for finitely many primes $p$, and positive integers $k_i,n_i > 0$. For a finite abelian group $A$, we denote by $d(A)$ the number of non-trivial summands in its primary decomposition.   We use $A_{(p)}$ to denote the subgroup of $A$ made of all elements of $A$ whose order is a power of $p$. Hence, 
\[
	A = \bigoplus_{p \text{ prime}} A_{(p)},\ \ \ \text{ where } A_{(p)} = \bigoplus_{k_i,n_i\,:\,i \in I_p} \left(\Z / p^{k_i} \Z \right)^{n_i}.
\] 

The following is a standard algorithmic result for finite abelian groups, formulated according to our needs in this article; see for example~\cite[Chap.~1.11]{Munkres84}.

\begin{theorem}
\label{thm:algoabeliangroup}
Let $H$ be a finite abelian group presented via a primary decomposition with generators $h_1, \ldots, h_{d(H)}$ for each of its prime summands. If $G \subseteq H$ is a subgroup of $H$ specified via a set of generators $g_1, \ldots, g_{d(G)}$ (each of which is encoded by a list of integers $n_{ij}$ such that that $g_i = \sum_j n_{ij} h_j$), then there exists a polynomial time algorithm using $O(d(|H|)^3)$ operations in $H$ that finds a presentation of the primary decomposition of $G$---that is, a new \emph{minimal} set of generators of $G$ denoted $g'_1, \ldots, g'_{d(G)}$ (each of which is, as before, represented by a list of integers $n'_{ij}$ such that $g'_i = \sum_j n'_{ij} h_j$) together with relations of the form $(g'_i)^{p_i^{k_i}}$ for primes $p_i$ and positive integers $k_i$. 
\end{theorem}

\paragraph{Pairings and forms on finite abelian groups}  We now quickly review the basics of quadratic forms on finite abelian groups.  We refer the reader to~\cite{Scharlau85} for more background.

Let $(A,+)$ be a finite abelian group, denoted additively. A {\em symmetric bilinear pairing} $b \colon A \times A \to G$ into the (not necessarily finite) abelian group $(G,+)$ is a map satisfying, for any $x,y,z \in A$, 
\[b(x,y)=b(y,x), \ \ \ b(x+y,z)=b(x,z) + b(y,z), \ \ \ \text{and,} \ \ b(x,y+z)=b(x,y) + b(x,z).\]
A pairing is {\em non-degenerate} if for every $x \in A$, $x \neq 0$, there exists some $y \in A$ such that $b(x,y)$ is not trivial in $G$. When the symmetric bilinear pairing takes value in $(G,+) = (\Q/\Z, +)$, where $(\Q/\Z,+)$ is the additive group of rational numbers modulo the integers, we call the pair $(A,b)$ a {\em discriminant pairing}.  For any real number $x \in \R$, denote by $\e(x)$ the quantity $\exp(2\pi i x)$. Let $(U(1), \times)$ be the multiplicative group of complex numbers $\{\e(x) : x \in \R/\Z\}$ on the unit circle. There is a natural group isomorphism $\Q/\Z \to U(1), x \mapsto \e(x)$. When a symmetric bilinear pairing $b$ is non-degenerate and takes value in $(U(1),\times)$, we call it a {\em bicharacter}. Note that $(A,b)$ is a non-degenerate discriminant pairing if and only if $(A,\e(b(\cdot,\cdot)))$ is a bicharacter.  
If $(A,b)$ is a discriminant pairing and $e_1, \ldots, e_k \in A$, then the {\em Gram matrix} of $(e_1, \ldots, e_k)$ 
is the $k \times k$ matrix $\left(b(e_i,e_j)\right)_{i,j}$. 

For a map $q \colon A \to G$, define $\partial q \colon A \times A \to G$ by $\partial q(x,y) := q(x+y)-q(x)-q(y)$ for any $x,y \in A$. The map $q$ is a {\em (homogeneous) quadratic form} if, for any $x \in A$ and $u \in \Z$, $q(u x) = u^2 q(x)$ and the map $\partial q$ is bilinear. For a quadratic form $q$, we call $\partial q$ the {\em bilinear pairing associated to $q$}. A quadratic form $q$ is {\em non-degenerate} if the bilinear pairing $\partial q$ is non-degenerate. If $q$ takes value in $\Q/\Z$, the pair $(A,q)$ is called a {\em pre-metric group}, and if moreover $q$ is non-degenerate, $(A,q)$ is a {\em metric group}.

\paragraph{Gauss sums on pre-metric groups} If $\gamma \in \Q/\Z$, let $\e(\gamma) := \exp(2\pi i \gamma)$.  To every pre-metric group $(G,q)$ we associate a complex number $\Theta(G,q)$ called the \emph{Gauss sum}:
\[
  \Theta(G,q) := \frac{1}{ | \sqrt{|G|} | } \sum_{x \in G} \e(q(x)).
\]
Our next theorem is one of the central tools in our algorithm for computing the TVBW invariants of Tambara-Yamagami categories. 

\begin{theorem}
\label{thm:evaluategausssum}
Let $G$ be a finite abelian group presented via a primary decomposition with $d(G)$ summands.  If oracle access to a quadratic form $q: G \to \Q/\Z$ is provided, then the Gauss sum $\Theta(G,q)$ can be evaluated in polynomial time using $O(d(G)^3 \poly(\log |G|))$ operations, $O(d(G)^2)$ oracle accesses, and $O(d(G)^2 \poly(\log |G|))$ memory words. 
\end{theorem}

Theorem \ref{thm:evaluategausssum} is seemingly folklore in the algebraic number theory literature, but we were unable to find any references proving it.   We understand it as a consequence of a scattered collection of works on pre-metric groups and algorithms.  
Wall~\cite{Wall63} abstractly classifies metric groups by showing how to decompose them into direct sums of elementary pieces; Basak and Johnson~\cite{BasakJ15} later provide an algorithmic formulation of the decomposition theorem and explain that the Gauss sum of a metric group can be computed readily from the decomposition. We have however not found a complexity analysis of this construction anywhere in the literature.  The closest we are aware of is \cite[Thm.~1.1]{cai2010tractable}, which shows that Gauss sums of (not-necessarily-homogeneous) quadratic forms on cyclic groups $\Z/N\Z$ are computable in polynomial time; however, even if $A$ itself is cyclic, the Gauss sums we will need to compute in our algorithm for $|M|_{\ty(A,\chi,\nu)}$ are on non-cyclic pre-metric groups of the form $(A^n, q)$. In order to fill this gap, we prove Theorem~\ref{thm:evaluategausssum} in Appendix~\ref{sec:gausssum} by supplementing the algorithm of~\cite{BasakJ15} with a detailed complexity analysis.

\subsection{Combinatorial and algebraic topology}
\label{ss:top}

\paragraph{(Co)homology}  Throughout this article we employ the simplicial cohomology of simplicial complexes with $\Z/2\Z$ coefficients.  We assume the reader is familiar with the basic definitions (see \cite{Hatcher02} for an introduction to the topic).  The following well-known fact can be understood as an application of Theorem \ref{thm:algoabeliangroup}.

\begin{theorem}
\label{thm:computabilitycohom}
Let $\tri$ be a simplicial complex with $n$ simplices. For any dimension $d = 0,1,2,\dots$, the following can be computed in polynomial time using $O(n^3)$ operations: a basis of $\Cycle^d(\tri,\Z/2\Z)$; a basis of $B^d(\tri;\Z/2\Z)$; a set of elements in $Z^d(\tri;\Z/2\Z)$ that represent a basis of $H^d(\tri;\Z/2\Z)$.  In particular, the dimension of each of these spaces can be computed in polynomial time.  \qed
\end{theorem}

\paragraph{Triangulated 3-manifolds}  The most common data type in computational 3-manifold topology is that of a generalized triangulation; e.g., see \cite{Matveev03}.  However, strictly speaking, the TVBW invariants are only defined for \emph{simplicial (or PL)} triangulations.  It is possible to extend the definition to ``polytope decompositions" (including generalized triangulations) \cite{kirillov2010turaev}, but we will not need to work at this level of generality.\footnote{Moreover, it is of no import for qualitative complexity theory, as the second barycentric subdivision of a generalized triangulation is always simplicial.}  Thus, throughout this work, we make the following conventions: a 3-manifold $M$ is always closed and oriented, and presented by a simplicial triangulation, denoted $\tri$, together with an orientation of $\tri$.  The orientability of $M$ is reflected in the fact that the simplicial homology group $H_3(\tri;\Z)$ is isomorphic to $\Z$ (this can be checked quickly from $\tri$), and then an \emph{orientation} of $\tri$ is encoded by labeling each tetrahedron in $\tri$ with a sign $+1$ or $-1$ in a way that gives rise to a $\{+1,-1\}$-valued cellular 0-cocycle on the cellulation dual to $\tri$.

Later, it will also be convenient to decorate each \emph{edge} of an oriented triangulation with an arbitrarily chosen orientation.  We stress that these edge orientations are completely arbitrary, and have nothing to do with the orientation of the manifold's tetrahedra.  Going forward, by a small abuse of language, if we say \emph{triangulation}, then we will usually mean an oriented triangulation equipped with such an edge orientation.

We denote the set of vertices of $\tri$ by $V$, the set of edges by $E$, the set of triangular faces by $F$, and the set of tetrahedra by $T$.

\subsection{Tambara-Yamagami quantum invariants}
\label{sec:tyquantuminvariant}

The {\em quantum invariants} we consider in this article are $\mathbb{C}$-valued topological invariants of closed 3-manifolds derived from algebraic data (a {\em spherical fusion category}) attached to a combinatorial presentation of the 3-manifold (a {\em simplicial triangulation}), together with evaluation rules described in the works of Turaev and Viro~\cite{TuraevViro} and Barrett and Westbury~\cite{BarrettWestbury}.  We now give an extremely brief introduction to {\em Tambara-Yamagami quantum invariants}, which are the TVBW invariants produced by a particular infinite family of spherical fusion categories called \emph{Tambara-Yamagami categories} \cite{TambaraY98}.  For the sake of space, here we provide only the details necessary to follow the proofs of our main results.  We refer the reader to Appendix \ref{sec:quantumtopology} for the full details of the combinatorial description of spherical fusion categories in general and for Tambara-Yamagami categories in particular, as well as the general construction of TVBW invariants (for unitary and multiplicity-free categories).

Let $(A,+)$ be a finite abelian group with $|A|$ elements, $\chi: A \times A \to U(1)$ a bicharacter on $A$, and $\nu$ a choice of square root $\nu = \pm1/ \sqrt{|A|} = \pm |A|^{-1/2}$. The {\em Tambara-Yamagami category} $\ty(A,\chi,\nu)$ is an algebraic object which gives rise to the following data:

\begin{itemize}
	\item a set of {\em colors} $L \stackrel{\text{def}}{=} A \sqcup \{m\}$ (where $m$ is some symbol not in $A$), together with a \emph{duality} operation $^*$ defined by $a^* \stackrel{\text{def}}{=} -a$ and $m^* \stackrel{\text{def}}{=} m$.
\item scalars $d_a$ and $d_m$ associated to every color in $L$ called {\em quantum dimensions}, as well as the {\em (global) quantum dimension} scalar $D$.  These scalars are defined as follows:
	\[ d_a \stackrel{\text{def}}{=} 1 \text{ for all } a \in A, \qquad d_m \stackrel{\text{def}}{=} +|A|^{1/2}, \qquad D \stackrel{\text{def}}{=} \sum_{x \in L} d_x^2 = 2|A|.\]
\item a set of {\em admissible triples} of colors
\[\{(a,b,-a-b),(a,m,m),(m,a,m),(m,m,a) : a,b \in A\} \subset L \times L \times L,\]
\item and a system of $\C$-valued {\em tetrahedron weights} defined below.
\end{itemize}



A {\em coloring} $\theta \colon E \to L$ of a triangulation $\tri$ is an assignment of a color to each edge of $\tri$.  Given an edge $e_0$ inside of a face $f$, the orientation of $e_0$ (recall from Section \ref{ss:top} that this data is included in $\tri$) determines a direction along which to traverse the three boundary edges of $f$; call these edges $e_0, e_1$ and $e_2$, according to the order in which we first see them while doing this traversal.  Now associate a triple of colors $(\theta(e_0),\theta(e_1)^{(*)},\theta(e_2)^{(*)})$, where $\theta(e_i)^{(*)}$ means we apply the duality operation $^*$ to the label $\theta(e_i)$ if the direction of traversal around $f$ disagrees with the orientation of $e_i$. 
The coloring $\theta$ is called {\em admissible} if every such triple associated to an edge-face pair is admissible. The set of {\it admissible colorings} is denoted by $\adm_L(\tri)$.

We now describe the tetrahedral weights involved in the definition of the Tambara-Yamagami invariants.  Let $\theta$ be a coloring of $\tri$ and let $\theta_e:=\theta(e) \in L$ denote the label of an edge $e \in E$. For each tetrahedron $t \in T$ with edges $e_1,e_2, \ldots e_6$ and orientation $O(t) \in \{+,-\}$, we assign a {\it tetrahedral weight} $|t|^{O(t)}_{\theta} \in \C$.  If $\theta$ is inadmissible, then $|t|^{O(t)}_{\theta} \stackrel{\text{def}}{=} 0$. Otherwise, $\theta$ is admissible and so, as described in Appendix \ref{sec:quantumtopology}, takes on one of four possible types; taking $O(t)$ into account we then assign the tetrahedral weights as follows:


\[
\begin{array}{cccc}
\text{{\bf $m$-empty}} &&& \text{{\bf $m$-triangle}}\\
\sixj{a}{b}{a+b}{b-c}{a+c}{c}^+ = 1 &&& \sixj{m}{m}{a}{a}{a+b}{m}^+ = d_m\\
\\
\sixj{a}{b}{a+b}{c-b}{a+c}{c}^- = 1 &&& \sixj{a}{b}{a+b}{m}{m}{m}^- = d_m \\
\\
\text{{\bf flat $m$-quad}} &&& \text{{\bf crooked $m$-quad}}\\
\sixj{a}{m}{m}{b}{m}{m}^+ = d_m \chi(a,b)  &&& \sixj{m}{m}{a}{m}{m}{b}^+ = \nu d_m^2 \overline{\chi(a,b)} \\
\\
\sixj{a}{m}{m}{b}{m}{m}^- = d_m \overline{\chi(a,b)}&&& \sixj{m}{m}{a}{m}{m}{b}^- = \nu^{-1} d_m^2\chi(a,b),\\
\end{array}
\]
where $a,b,c$ are arbitrary elements in $A$.

Finally, the Tambara-Yamagami quantum invariant $|M|_{\ty(A,\chi,\nu)}$ of a 3-manifold $M$ with triangulation $\tri$ is defined via the following ``state sum formula:" 
\[
	|M|_{\ty(A,\chi,\nu)} \stackrel{\text{def}}{=} 
  \displaystyle\sum_{\theta \in \adm_L(\tri)} 
    \frac{ 
      \displaystyle\prod_{t \in T} 
      |t|_\theta^{O(t)}
      \prod_{e \in E} d_{\theta{e}} 
      }
      {
      \displaystyle\prod_{f \in F} \sqrt{d_{\theta{e_1}}d_{\theta{e_2}}d_{\theta{e_3}}} \prod_{v \in V}D.
      }
\]
We note that for a fixed coloring $\theta$ and tetrahedron $t$, testing its admissibility on $t$ and determining its type can be carried out in constant time.  Thus, there is a na\"ive algorithm to compute $|M|_{\ty(A,\chi,\nu)}$ directly from this definition with running time exponential in the size of $\tri$. In fact, this is true of $|M|_{\calC}$ for any spherical fusion category $\calC$. We now turn to showing that we can do much better when $\calC=\ty(A,\chi,\nu)$ if we bound the $\Z/2\Z$ Betti number $\beta_1(M)$.

\section{The role of the Betti number in Tambara-Yamagami invariants}
\label{sec:fptalgorithm}

We are interested in the following algorithmic problem:
\begin{center}
\begin{tabular}{|l|}
\hline
{\sc $\ty(A,\chi,\nu)$-invariant computation}\\
\hline
{\bf Input:} $\tri$ a triangulation of an oriented closed 3-manifold $\M$\\
{\bf Output:} $|\M|_{\ty(A,\chi,\nu)}$\\
\hline
\end{tabular}
\end{center}
\noindent
Since we are not working uniformly in the data $A,\chi,\nu$, there is no harm in assuming we have full ``explicit'' access to it.\footnote{Nevertheless, we note that all we really need is access to the primary decomposition of $A$ and oracle access to a function $(x,y) \mapsto b(x,y) \in \Q/\Z$ such that $\chi(x,y) = \e(b(x,y))$, for any $x,y \in A$.  In fact, using \cite[Thm.~1.1]{cai2010tractable}, it should be possible to relax the former in the case that $A=\Z/N\Z$ is cyclic.  In other words, it should not be necessary to assume we know the prime factorization of $N$.}
Our main result is 
\begin{theorem}
\label{thm:maincpx}
For a fixed Tambara-Yamagami category $\ty(A,\chi,\nu)$, there is a deterministic algorithm that, given as input a triangulation of an orientable closed 3-manifold $\M$ with $n$ tetrahedra, computes the quantum invariant $|\M|_{\ty(A,\chi,\nu)}$ in $O( 2^{\beta_1} n^3 )$ operations and using $O(n^2)$ memory words. 
\end{theorem}

In this section, we make precise the ideas described in Section \ref{sec:introduction} that lead to such an algorithm.  We put everything together and prove Theorem \ref{thm:maincpx} in Section \ref{sec:complexity}.  For the rest of this section and the next, $A, \chi,$ and $\nu$ are fixed.

\subsection{Admissible colorings induce simplicial 1-cocycles}
\label{ss:cocyle}
We begin by noting that each admissible coloring determines a $\Z/2\Z$-valued simplicial 1-cocycle.  More precisely, consider the function
\begin{equation}
\label{eq:projphi}
  \phi\colon \adm_L(\tri) \to \Chain^1(\tri, \Z/2\Z), \ \ \ \ 
              \phi(\theta)(e) = \left\{ \begin{array}{ll}
                                       1 & \text{if } \theta(e) = m\\
                                       0 & \text{otherwise}\\
                                       \end{array} \right.
\end{equation}
that assigns to an admissible coloring of $\tri$ a $\Z/2\Z$-valued simplicial 1-cochain of the triangulation (that is, a function $E \to \Z/2\Z$). 

\begin{lemma}
\label{lem:onecocycle}
For any admissible coloring $\theta$ of a triangulation $\tri$, $\phi(\theta)$ is a 1-cocycle. \qed
\end{lemma}

In the case of $\ty(\Z/2\Z,\exp(i\pi a b), -1/\sqrt{2}) = \tv_4$, this lemma was proved in \cite{BurtonMS18,MariaS16} and used in the algorithm of \cite{MariaS20}. It can be proved easily in the general case by examining the definition of admissible triples in $\ty(A,\chi,\nu)$---the calculation boils down to the simple observation that in any admissible coloring, the number of edges around any face of $\tri$ colored by $m$ is either 0 or 2.  More conceptually, this lemma follows from the fact that $\ty(A,\chi,\nu)$ is not just a spherical fusion category---it is a spherical \emph{$\Z/2\Z$-graded} fusion category, and a similar result holds for any spherical $G$-graded fusion category, cf.~\cite{turaev20123}.

\subsection{The partial state sum at a 1-cocycle} 
\label{ss:partial}
Fix a 1-cocycle $\alpha \in Z^1(\tri,\Z/2\Z)$ and consider the {\em partial state sum at $\alpha$}, defined by only summing over the admissible labelings associated to $\alpha$: 
\[
	|\tri,\alpha|_{\ty(A,\chi,\nu)} \stackrel{\text{def}}{=}
  \displaystyle\sum_{\theta \in \phi^{-1}(\alpha)} 
    \frac{ 
      \displaystyle\prod_{t \in T} |t|^{O(t)}_\theta
      \prod_{e \in E} d_{\theta{e}} 
      }
      {
      \displaystyle\prod_{f \in F} \sqrt{d_{\theta{e_1}}d_{\theta{e_2}}d_{\theta{e_3}}} \prod_{v \in V}D
      }.
\]

By definition of $\phi$, all admissible colorings of $\phi^{-1}(\alpha)$ have the same set of edges colored by the object $m$. We partition edges, faces, and tetrahedra of $\tri$, such that for every coloring in $\phi^{-1}(\alpha)$:
\begin{enumerate}
\item $E = \Eemp \sqcup \Eone$ where $\Eemp$ is the set of edges colored by an element of $A$, and $\Eone$ is the set of edges colored by $m$,
\item $F = \Femp \sqcup \Ftwo$ where $\Femp$ is the set of faces whose edges are colored exclusively by elements of $A$, and $\Ftwo$ is the set of faces with exactly two edges colored by $m$,
\item $T = \Temp \sqcup \Ttriang \sqcup \Tquadplus \sqcup \Tquadminus \sqcup \Tcrookquadplus \sqcup \Tcrookquadminus$, where $\Temp$ is the set of $m$-empty faces, $\Ttriang$ is the set of $m$-triangles, $\Tquadplus$ and $\Tquadminus$ are respectively the set of positively and negatively oriented flat $m$-quads, and $\Tcrookquadplus$ and $\Tcrookquadminus$ are respectively the set of positively and negatively oriented crooked $m$-quads (see end of Section~\ref{sec:tyquantuminvariant} for conventions). 
\end{enumerate}

Replacing $|t|^{O(t)}_\theta$ by the appropriate values and factorizing, we now rewrite the partial state sum at the 1-cocycle $\alpha$. 
%
For an admissible coloring $\theta$ and an $m$-quad tetrahedron $t$ in $\theta$ (either flat or crooked), define $\chi_{\theta(t)}$ to be the value $\chi(a,b)$ where $a,b$ are the colors in $\theta$ assigned to the two opposite edges not colored by $m$ in the (flat or crooked) $m$-quad tetrahedron $t$.  With this notation, one can check:
\begin{equation}
\label{eq:partialstatesum}
	|\tri,\alpha|_{\ty(A,\chi,\nu)} =
    D^{-|V|} \cdot d_m^{C_\alpha} \cdot 
	\nu^{|\Tcrookquadplus|-|\Tcrookquadminus|}  
	\sum_{\theta \in \phi^{-1}(\alpha)}
	\left(
    \prod_{t \in \Tquadplus} {\chi_{\theta(t)}} \cdot
    \prod_{t \in \Tquadminus} \overline{\chi_{\theta(t)}} \cdot
    \prod_{t \in \Tcrookquadplus} \overline{\chi_{\theta(t)}} \cdot
    \prod_{t \in \Tcrookquadminus} {\chi_{\theta(t)}}
    \right)
    \end{equation}
where:
\[ C_\alpha = |\Eone|-|\Ftwo|+|\Tquadplus|+|\Tquadminus| + 2|\Tcrookquadplus|+2|\Tcrookquadminus|.\]

In Section~\ref{subsec:evaluation} we prove:
\begin{theorem}
\label{thm:polytime_subroutine}
Fix a category $\ty(A,\chi,\nu)$. Let $e_1, \ldots, e_{|E|}$ be an arbitrary ordering of the edges of $\tri$. For any collections $\calP, \calQ$ of pairs of indices $(i,j)$, $1 \leq i,j \leq |E|$, with potentially repeated pairs, such that $|\calP|,|\calQ| \in O(n)$, there 
is an algorithm to evaluate the sum:
\begin{equation}
\label{eq:genericpartialstatesum}
\sum_{\theta \in \phi^{-1}(\alpha)} \prod_{(i,j) \in \calP} \chi(\theta_{e_i},\theta_{e_j}) \prod_{(r,s) \in \calQ} \overline{\chi(\theta_{e_r},\theta_{e_s})}
\end{equation}
in $O(n^3)$ operations, using $O(n^2)$ memory words.
\end{theorem}
\noindent
Using the theorem, we arrive at a key step towards our main result.

\begin{corollary}
\label{cor:polytimesubroutine}
Fix a category $\ty(A,\chi,\nu)$. Given as input a 3-manifold triangulation $\tri$ with $n$ tetrahedra, and a 1-cocycle $\alpha$ in $\tri$, the sum~(\ref{eq:partialstatesum}) can be evaluated in $O(n^3)$ operations, using $O(n^2)$ memory words.
\end{corollary}

\begin{proof}
By definition of the projection $\phi$ (Equation~(\ref{eq:projphi})), the 1-cocycle $\alpha$ gives a description of the edges colored by the object $m$ in all admissible colorings of $\phi^{-1}(\alpha)$. Additionally, the multiplicative constant $D^{-|V|} \cdot d_m^{C_\alpha} \cdot \nu^{|\Tcrookquadplus|-|\Tcrookquadminus|}$ in front of the sum (Equation~(\ref{eq:partialstatesum})) depends exclusively on the set of edges colored by $m$. In consequence, the multiplicative constant can be evaluated in linear time by checking the set of $m$-colored edges of each tetrahedron. Using the notations of Theorem~\ref{thm:polytime_subroutine}, we consider the pairs $\calP$ (resp. $\calQ$) of indices $(i,j)$ of opposite edges $(e_i,e_j)$ not colored by $m$ in tetrahedra of type $\Tquadplus \sqcup \Tcrookquadminus$ (resp. $\Tquadminus \sqcup \Tcrookquadplus$), in order to evaluate the sum $\sum_{\theta \in \phi^{-1}(\alpha)}
  \left(
    \prod_{t \in \Tquadplus} {\chi_{\theta(t)}} \cdot
    \prod_{t \in \Tquadminus} \overline{\chi_{\theta(t)}} \cdot
    \prod_{t \in \Tcrookquadplus} \overline{\chi_{\theta(t)}} \cdot
    \prod_{t \in \Tcrookquadminus} {\chi_{\theta(t)}}
    \right)$ in cubic time. The corollary follows.
\end{proof}

\subsection{Admissible colorings over a 1-cocycle can be efficiently characterized as a finite abelian group}
\label{sec:phiminusoneasagroup}
For a given 1-cocycle $\alpha \in Z^1(\tri;\Z/2\Z)$, we now characterize the set $\phi^{-1}(\alpha) \subseteq \adm(\tri)$. Recall that $\Eemp$ is the set of edges not colored by $m$ and $\Femp$ is the set of faces with no edge colored by $m$.


\begin{lemma}
\label{lem:kernel}
If the finite abelian group $A$ is given by its primary decomposition with $d(A)$ summands, and the triangulation $\tri$ has $n=|\tri|$ tetrahedra, then the set $\phi^{-1}(\alpha)$ is in bijection with the elements of a finite abelian group $G$ with $O(d(A)n)$ generators. 
\end{lemma}

\begin{proof}
Fix an arbitrary orientation of the edges of the triangulation. The admissibility constraints for colorings in a Tambara-Yamagami category imply that the set $\phi^{-1}(\alpha)$ is in bijection with the set of assignments $\theta \colon \Eemp \to A$ satisfying the following linear equations: for any $m$-empty face with edges $e_1,e_2,e_3$, 
$$\varepsilon_1\theta(e_1) + \varepsilon_2\theta(e_2) + \varepsilon_3\theta(e_3) = 0,$$ where $\varepsilon_i \in \{\pm1\}$ depending whether the orientations of edges $e_1,e_2,e_3$ agree or not with an arbitrarily chosen orientation of the triangular face. 

Consider the group homomorphism:
\[
  \lambda\colon A^{|\Eemp|} \to A^{|\Femp|}
\]
where the $|\Femp|$-many values are the relations $\varepsilon_1\theta(e_1) + \varepsilon_2\theta(e_2) + \varepsilon_3\theta(e_3)$ above. 


The set $\phi^{-1}(\alpha)$ is in bijection with the kernel of this application. Because $A^{|\Eemp|}$ is a finitely generated abelian group with $d(A) |\Eemp|$ summands in its primary decomposition, with $|\Eemp| \leq |E| \in O(n)$, then $\ker \lambda$ is a finitely generated abelian group with $O(d(A)n)$ generators. 
\end{proof}

Additionally, a minimal family of generators $(y_1, \ldots, y_\ell), y_i \in A^{|\Eemp|}$, for $\ker \lambda$ can be computed in polynomial time.

\begin{lemma}
\label{lem:minfamgeneratorskerl}
For a fixed abelian group $A$, with the notations above, a minimal family of generators for $\ker \lambda$ can be computed in $O(n^3)$ operations, using $O(n^2)$ memory words.
\end{lemma}

\begin{proof}
A minimal family of generators of $\ker \lambda$ can be computed using a Smith normal form algorithm on the integer matrix representing the homomorphism $\lambda$ on a canonical set of generators of products of $A$s. 

Write $d$ for $d(A)$, and let $x_1, \ldots, x_d$ be a fixed family of generators for $A$. Denote by $(x_{ij})_{i=1\ldots d, j=1 \ldots k}$ a family of generators for $A^k$, where the subfamily $(x_{ij})_{i=1\ldots d}$ are the generators $x_1, \ldots, x_d$ for the $j$th copy of $A$ in $A^k$. Do so for $A^{|\Eemp|}$ and $A^{|\Femp|}$. Note that in our framework, $A$ is given by its primary decomposition, and we get automatically a canonical set of generators as required.

Let $M$ be the $(d|\Femp|) \times (d|\Eemp|)$ integer matrix representing the group morphism $\lambda$ expressed in the above set of generators. By definition of $\lambda$, it is a matrix with coefficients in $\{-1,0,+1\}$. Let $D = V M U$ be its Smith normal form, where $U$ and $V$ are unimodular, and $D$ is a diagonal matrix, with $D_{ii} \neq 0$ \ifff $i = 1 \ldots |\Eemp|-\ell$, and $\ell$ is the dimension of the kernel of $\lambda$. Such decomposition can be computed in $O(n^3)$ operations in $\Z$ with Gaussian eliminations, considering $d(A)$ constant and $|\Eemp|, |\Femp| \in O(n)$.

The rightmost $\ell$ columns of matrix $U$ encode, in the $(x_{ij})$ set of generators of $A^{|\Eemp|}$, a minimal set of generators for the kernel of $\lambda$.

To prevent an arithmetic blow-up in the coefficients during the reduction, all arithmetic operations can be performed in $\Z/|A|\Z$, considering that any group element in $\ker \lambda$ has order at most $|A|$. For a fixed group $A$, coefficients are constant size and arithmetic operations are constant time (they are bounded by small polynomials in $\log |A|$).
%
%
%
\end{proof}

In conclusion, any generator of the kernel $\ker \lambda \subseteq A^{|\Eemp|}$ can be interpreted as an assignment of colors to the empty edges of the triangulation $\Eemp \to A$. For a fixed 1-cocycle $\alpha$ of the triangulation, any admissible coloring in $\phi^{-1}(\alpha)$ can be expressed as a sum (with integer coefficients) of elements of the kernel (where the sum of colors is edge-wise) following the rule of addition in $A$.


\subsection{Partial state sums are (normalized) Gauss sums} 
\label{subsec:evaluation}
The goal of this subsection is to prove Theorem \ref{thm:polytime_subroutine}.
Consider a finite abelian group $(A,+)$ with $d := d(A)$ generators $x_1, \ldots, x_d$ for the summands of its primary decomposition, its product $A^k$ with generators $(x_{ij})_{i=1\ldots d, j=1\ldots k}$, and a subgroup $(G,+)$ with $g := d(G)$. 

For a group element $y \in G$, we express $y = \sum_{i,j} n_{ij}(y) x_{ij}$, for integer coefficients $n_{ij}(y)$, in the set of generators for $A^k$. We also write $y^{(j)} := \sum_{i = 1 \ldots d} n_{ij}(y) x_i \in A$ for the projection of $y$ into the $j$-th component of the product $A^k$.

Let $\calP,\calQ$ be arbitrary sets of pairs of indices $(i,j)$, with $1 \leq i,j \leq k$, and possible multiplicities $p_{ij} \geq 1$ for the pair $(i,j)\in \calP$, and $q_{ij}$ for the pair $(i,j) \in \calQ$.
Consider the product: 


\begin{equation}
\label{eq:form}
  \psi \colon G \to U(1), \ \ \ \  
  y \mapsto \displaystyle \prod_{(i,j)\in \calP} \chi(y^{(i)},y^{(j)})^{p_{ij}}\prod_{(r,s)\in \calQ} \chi(y^{(r)},y^{(s)})^{-q_{rs}}.
\end{equation}
and define the complex valued sum:
\begin{equation}
\label{eq:gausssum}
\Omega(\calP,\calQ) := \displaystyle\frac{1}{|\sqrt{|G|}|}\sum_{y \in G} \psi(y).
\end{equation}

\begin{lemma}
\label{lem:quadraticgausssum}
The sum $\Omega(\calP,\calQ)$ in Equation~(\ref{eq:gausssum}) is a {\em quadratic Gauss sum} over $G$, \ie, there exists a quadratic form $q \colon G \to \Q/\Z$ s.t.,
\[
  \Omega(\calP,\calQ) = \Theta(G,q) = \displaystyle\frac{1}{|\sqrt{|G|}|}\sum_{x \in G} \e(q(x)).
\]
\end{lemma}

\begin{proof}
Define the pairing $b \colon A \times A \to \Q/\Z$ satisfying, for any $x,y \in A$, the equality $\chi(x,y) = \e(b(x,y))$. 
The map $b$ is a bilinear pairing. This follows from the definition of a bicharacter. For any $x,y,z \in A$,
\[
  \chi(x+y,z) = \e(b(x+y,z)) = \chi(x,z)\chi(y,z) = \e(b(x,z) + b(y,z)),
\]
and in consequence $b(x+y,z) = b(x,z)+b(y,z) \mod \Z$. By symmetry of $\chi$, bilinearity follows.

\medskip

Define $q \colon G \to \Q/\Z$ by:
\[
  q(y) := \sum_{(i,j)\in \calP} p_{ij} \cdot b(y^{(i)},y^{(j)}) - \sum_{(r,s)\in \calQ} q_{rs} \cdot b(y^{(r)},y^{(s)}),
\]
such that $\left|\sqrt{|G|}\right| \Omega(\calP,\calQ) = \sum_{x \in G} \e(q(x))$. The map $q$ is a quadratic form. 
Indeed, first notice that, for any integer $u$, $q(u y) = u^2 q(y)$ follows from the fact that $b(u y^{(i)}, u y^{(j)}) = u^2 b(y^{(i)}, y^{(j)})$. Second, the form:
\[
  \partial q : (y,y') \mapsto q(y + y')-q(y)-q(y')
\]
is bilinear, as it can be expressed as a sum of evaluations of the bilinear pairing $b$ on $(y^{(i)}, y'^{(j)})$ for some $i,j$. 
\end{proof}

\begin{proof}[Proof of Theorem~\ref{thm:polytime_subroutine}]
Using the notations of Theorem~\ref{thm:polytime_subroutine}, we identify the following sum on the triangulation:
\begin{equation}
\label{eq:inproofgsum}
\sum_{\theta \in \phi^{-1}(\alpha)} \prod_{(i,j) \in \calP} \chi(\theta_{e_i},\theta_{e_j}) \prod_{(r,s) \in \calQ} \overline{\chi(\theta_{e_r},\theta_{e_s})}
\end{equation}
with the quantity $\left|\sqrt{|G|}\right| \Omega(\calP,\calQ)$ in Equation~(\ref{eq:gausssum}), where the group $G$ is the subgroup $\ker \lambda$ of the product $A^{|\Eemp|}$. The pairs of indices $\calP$ and $\calQ$ represent indices of edges in Equation~(\ref{eq:inproofgsum}), and particular copies of $A$ in the product $A^{|\Eemp|}$ in Equation~(\ref{eq:form}), with a natural correspondence induced by the fact copies of $A$ in $A^{|\Eemp|}$ are in bijection with edges not colored by $m$ (once the 1-cocycle $\alpha$ is fixed). 

The number of elements of the finite group $G = \ker \lambda$ can be evaluated by finding the primary decomposition of $G$, which can be done in polynomial time (Theorem~\ref{thm:algoabeliangroup}). Let $q$ be the quadratic form defined in the proof of Theorem~\ref{lem:quadraticgausssum}. We can compute the Gram matrix of a minimal family of generators of $\ker \lambda$ (computed in Lemma~\ref{lem:minfamgeneratorskerl}) for the discriminant pairing $(G,\partial q)$ by evaluating $\partial q$, and compute the Gauss sum $\Omega(\calP,\calQ)$ in polynomial time (by virtue of Theorem~\ref{thm:evaluategausssum}), with the expected complexity.
\end{proof}

\subsection{Cohomologous 1-cocycles have equal partial state sums}
\label{ss:cohomologous}
The tools introduced so far in this section would only be enough to construct a parameterized algorithm in terms of the number of 1-cocycles $Z^1(\tri;\Z/2\Z)$.  The final ingredient necessary to further reduce the parameter to $\beta_1(M)$ follows from a general statement about TVBW invariants of spherical $G$-graded fusion categories (for any finite group $G$) first described by Turaev and Virelizier \cite[Thm.~7.1]{turaev20123}, specialized to our situation: $\ty(A,\chi,\nu)$ is a spherical $\Z/2\Z$-graded fusion category.

\begin{theorem}[special case of Thm.~7.1 in \cite{turaev20123}]
	\label{thm:cohomologous}
	If $\alpha$ and $\beta$ are cohomologous elements of $Z^1(\tri;\Z/2\Z)$, then
	$|\tri,\alpha|_{\ty(A,\chi,\nu)} = |\tri,\beta|_{\ty(A,\chi,\nu)}$. \hfill \qed
\end{theorem}

\section{Fixed parameter tractable algorithm in the first Betti number}
\label{sec:complexity}
Fix the data of $\ty(A,\chi,\nu)$.  Given as input an arbitrary triangulation $\tri$ of a closed oriented 3-manifold $\M$, our algorithm to compute the TVBW invariant $|M|_{\ty(A,\chi,\nu)}$ is as follows:

\medskip

\begin{enumerate}




%
\item\label{itm:cohomology} Compute a set of 1-cocycles $\{\alpha_1, \ldots, \alpha_{\beta_1}\} \subset \Cycle^1(\tri, \Z/2\Z)$ such that $\{[\alpha_1], \ldots, [\alpha_{\beta_1}]\}$ forms a basis of the 1-cohomology vector space $H^1(\tri, \Z/2\Z)$.

\smallskip

Compute a basis for $B^1(\tri, \Z/2\Z)$ ; the cardinality $\#B^1(\tri, \Z/2\Z)$ of the set of 1-coboundaries is equal to $2^{\dim B^1(\tri, \Z/2\Z)}$. \hfill {\footnotesize [Theorem~\ref{thm:computabilitycohom}]}

\medskip

\item\label{itm:enumerate} 
For each of the $2^{\beta_1}$ distinct 1-cohomology classes $[\alpha]$ of $H^1(\tri, \Z/2\Z)$, construct a representative $\alpha \in \Cycle^1(\tri, \Z/2\Z)$ by enumerating all $2^{\beta_1}$ $\Z/2\Z$-combinations of elements in $\{\alpha_1, \ldots, \alpha_{\beta_1}\}$, \ie, all $\alpha = \varepsilon_1 \alpha_1 + \ldots + \varepsilon_{\beta_1}\alpha_{\beta_1}$, for $\varepsilon_i \in \Z/2\Z$. 
  
\smallskip

Now, for each such $\alpha$, apply Corollary~\ref{cor:polytimesubroutine} to compute the partial state sum $|\M,\alpha|_{\ty(A,\chi,\nu)}$, with the following steps:


\begin{enumerate}
  \item\label{itm:SNF} Compute a minimal family of generators for $\ker \lambda = \langle x_1, \ldots, x_g\rangle$, \hfill {\footnotesize [Lemma~\ref{lem:minfamgeneratorskerl}]}

  \smallskip

  \item\label{itm:decompose} Compute the primary decomposition of $\ker \lambda$ with generators $\ker \lambda = \langle x'_1, \ldots, x'_g\rangle$ for the summands, \hfill {\footnotesize [Theorem~\ref{thm:algoabeliangroup}]}

  \smallskip

  \item\label{itm:gram} Compute the Gram matrix of $x'_1, \ldots, x'_g$ for the pre-metric form $(\ker \lambda, q)$, by evaluating $\partial q$ on all pairs $(x'_i,x'_j)$. \hfill {\footnotesize [defined in Lemma~\ref{lem:quadraticgausssum}]}

  \smallskip

  \item\label{itm:gauss} Evaluate the Gauss sum of $(\ker \lambda, q)$, \hfill {\footnotesize [Theorem~\ref{thm:evaluategausssum}]} 

\smallskip

  and multiply by $\left|\sqrt{|G|}\right|$. \hfill {\footnotesize [proof of Theorem~\ref{thm:polytime_subroutine}]}

\smallskip

  Normalize to get the state sum $|\M,\alpha|_{\ty(A,\chi,\nu)}$, with multiplicative factor as in Equation~(\ref{eq:partialstatesum}). 

  \hfill {\footnotesize [Corollary~\ref{cor:polytimesubroutine}]}

\smallskip

\end{enumerate}


  \item\label{itm:sum} Sum over all $\alpha$:
  \[
  	|\M|_{\ty(A,\chi,\nu)} = 
    \sum_{\begin{array}{c}
            \alpha = \varepsilon_1 \alpha_1 + \ldots \varepsilon_{\beta_1} \alpha_{\beta_1}\\
            \text{with } \varepsilon_i \in \Z/2\Z, \text{ and s.t.,} \\
            \langle [\alpha_1], \ldots, [\alpha_{\beta_1}]\rangle = H^1(\tri,\Z/2\Z)\\ 
          \end{array}
          } \#B^1(\tri,\Z/2\Z) \cdot |\M,\alpha|_{\ty(A,\chi,\nu)}.
  \]
\end{enumerate}

Our main Theorem~\ref{thm:maincpx} follows readily from the following:

\begin{theorem}
Fix a Tambara-Yamagami category $\ty(A,\chi,\nu)$. Given a triangulation $\tri$ of a closed 3-manifold $\M$ with $n$ tetrahedra, the above algorithm computes $|\M|_{\ty(A,\chi,\nu)}$ in $O(2^{\beta_1} n^3)$ operations and using $O(n^2)$ memory words.
\end{theorem}

\begin{proof}

Computing a basis for $B^1(\tri, \Z/2\Z)$, as well as a set of generators $\{\alpha_1, \ldots, \alpha_{\beta_1}\}$ whose cohomology classes form a basis of $H^1(\tri,\Z/2\Z)$, in step~(\ref{itm:cohomology}) is a standard procedure in computer algebra (Theorem~\ref{thm:computabilitycohom}), done by normalizing a $O(|\tri|)\times O(|\tri|)$ matrix with $\Z/2\Z$ coefficients. This can be done in $O(|\tri|^3)$ operations in $\Z/2\Z$ with Gaussian eliminations, where an arithmetic operation in $\Z/2\Z$ has constant complexity. 


Given the 1-cocycles $\{\alpha_1, \ldots, \alpha_{\beta_1}\}$, we can enumerate every combination $\alpha = \varepsilon_1 \alpha_1 + \ldots + \varepsilon_{\beta_1}\alpha_{\beta_1}$ in (amortized) $O(n)$ operations per $\alpha$. Indeed, enumerating all $\{\varepsilon_1, \ldots, \varepsilon_{\beta_1}\}$ is akin to incrementing a $\beta_1$-bits binary counter from $0 \ldots 0$ to $1 \ldots 1$, which induces an (amortized) change of $O(1)$ bits per increment~\cite[Chapter 16]{Cormen:2001:IA:580470}. The 1-cocycles are represented by formal sums of $O(n)$ edges, and computing the sum of two 1-cocycles takes $O(n)$ operations.

For each $\alpha$ as above, we compute the partial state sum using Corollary~\ref{cor:polytimesubroutine} in polynomial time. Specifically, step~(\ref{itm:SNF}) is solved by computing a Smith normal form, together with transformation matrices, of a particular $O(d(A)n) \times O(d(A)n)$-integer matrix, which has complexity $O(n^3)$ (Lemma~\ref{lem:minfamgeneratorskerl}) with (modular) Gaussian eliminations, once $A$ is fixed. 

We can read-off the generators $x_1,\ldots,x_g$ of $\ker \lambda$ from the transformation matrices, and compute the generators $x'_1,\ldots,x'_g$ of step~(\ref{itm:decompose}) by computing the primary decomposition of the subgroup $\ker \lambda$ (Theorem~\ref{thm:algoabeliangroup}).

Computing the Gram matrix in step~(\ref{itm:gram}) is done by evaluating the bilinear pairing $\partial q$, defined in Lemma \ref{lem:quadraticgausssum}, on all pairs $(x'_i,x'_j)$. The quadratic form $q$ is defined as a sum of $O(n)$ terms of the form $b(x'_i,x'_j)$, and the overall Gram matrix can be computed in $O(g^2n)$ operations. Recall that we have an oracle access to the function $b(\cdot,\cdot)$ (see beginning of Section~\ref{sec:fptalgorithm}).

By Theorem~\ref{thm:evaluategausssum}, the Gauss sum of $(\ker \lambda, q)$ can be evaluated from the Gram matrix of $x'_1, \ldots, x'_g$ in $O(g^3)$ operations in $\Q/\Z$, where the arithmetic complexity of operations in $\Q/\Z$ depend solely on the maximal order of an element of $A$ (independently of the size of the Gram matrix). See Appendix~\ref{sec:gausssum} for details. 

Finally, the invariant $|\M|_{\ty(A,\chi,\nu)}$ computed in step~(\ref{itm:sum}) is the sum of $2^{\beta_1}$ complex numbers represented with a constant number of memory words each for a fixed group $A$ (with $\beta_1 \leq n$). 


In consequence, the overall complexity of the algorithm, knowing $g \in O(d(A)n)$ and considering $d(A), |A|$ constant for a fixed group $A$, is $O(2^{\beta_1}n^3)$.
The memory complexity is bounded by the size of the matrices, which is $O(n^2)$, since we have a fixed $A$ (coefficients are represented with a constant number of memory words, and $d(A)$ is a constant).
\end{proof}

\section*{Acknowledgments}
CD was partially supported through the Simons Collaboration on Global Categorical Symmetries (Simons Foundation Award 888988). CM was partially supported by the ANR project ANR-20-CE48-0007 (AlgoKnot). ES was partially supported by NSF award DMS-2038020.

\bibliographystyle{alpha}
\bibliography{bibliography}

\newpage

\appendix

\section{Complexity of evaluating Gauss sums on finite abelian groups}
\label{sec:gausssum}

In this section, with prove the algorithmic complexity bounds of Theorem~\ref{thm:evaluategausssum}. We first introduce further background on pre-metric groups and their Gauss sums. Two subgroups $G_1,G_2$ of $G$ are {\em orthogonal} w.r.t. to a bilinear pairing $b$ if, for all $x \in G_1$ and $y \in G_2$, $b(x,y)=0$. Notably:

\begin{lemma}
\label{lem:orthopq}
Any two $p$-subgroup $G_{(p)}$ and $q$-subgroup $G_{(q)}$ of a group $G$, for $p,q$ distinct primes, are orthogonal for any bilinear pairing on $G$. 
\end{lemma}

A discriminant pairing $(G,b)$ admits an {\em orthogonal decomposition}, denoted by $(G_1,b) \perp (G_2,b)$, \ifff:
\[
  G_1,G_2 \neq 0, \ \ G = G_1 \oplus G_2, \ \ b(x,y) = 0 \ \text{whenever} \ x\in G_1 \ \text{and} \ y \in G_2.
\]  

A discriminant pairing which admits no non-trivial orthogonal decomposition is called {\em irreducible}. We say that a metric group $(G,q)$ admits an orthogonal decomposition if $(G,\partial q)$ admits one.


Gauss sums split multiplicatively under orthogonal decomposition:

\begin{lemma}
\label{lem:gaussusummult}
Gauss sums are multiplicative with respect to orthogonal sums, \ie, for $(G,q) = (G_1,q_1) \perp (G_2,q_2)$, we have:
\[
  \Theta(G,q) = \Theta(G_1,q_1)\Theta(G_2,q_2).
\]
\end{lemma}

Mainly due to the work of Wall~\cite{Wall63}, irreducible metric groups are classified, as well as the value of their Gauss sums~\cite{BasakJ15}: 

\begin{theorem}[\cite{Wall63}, after \cite{BasakJ15}]
\label{thm:metricgroupdec}
Let $(G,q)$ be a metric group. It admits a unique orthogonal decomposition, up to permutation of the terms, 
\[
  (G,q) = (G_1,q_1) \perp \ldots \perp (G_k,q_k)
\]
into irreducible metric groups $(G_i,q_i)$ of the following types:
\[
\arraycolsep=10pt\def\arraystretch{2}
  \begin{array}{|c|c|c|c|}
  \hline
  \text{G} & q(x) \in \Q/\Z & \Theta(G,q) & \left(\partial q\right) \\
  \hline
  \displaystyle\Z/p^r\Z ,p>2,r>0 & \displaystyle\alpha \frac{(p^r+1)}{2p^r} \cdot x^2, \alpha \in \{1, u_p\} & \displaystyle\left(\frac{2\alpha}{p}\right)^r \epsilon_{p^r} & \displaystyle\frac{\alpha}{p^r}\\
  \hline
  \displaystyle\Z/2^r\Z & \displaystyle\frac{\alpha}{2^{r+1}} \cdot x^2, \alpha \in \{-5,-1,1,5\} & \displaystyle(-1)^{r(\alpha^2-1)/8}\e(\alpha/8) & \displaystyle\frac{\alpha}{2^r}\\
  \hline
  \displaystyle(\Z/2^r\Z)^2 & \displaystyle\frac{x_1x_2}{2^{r}} & 1 & 
  \displaystyle\left(\arraycolsep=1pt\def\arraystretch{1}\begin{array}{cc}
                                            0 & 2^{-r}\\ 2^{-r} & 0 \\
                                            \end{array}\right)  \\
  \hline 
  \displaystyle(\Z/2^r\Z)^2 & \displaystyle\frac{x_1^2+x_1x_2+x_2^2}{2^{r}} & (-1)^r & \displaystyle\left(\arraycolsep=1pt\def\arraystretch{1}\begin{array}{cc}
                                            2^{1-r} & 2^{-r}\\ 2^{-r} & 2^{1-r} \\
                                            \end{array}\right)  \\
  \hline
  \end{array}
\]
Here, $u_p$ is a quadratic non-residue mod $p$ and $\epsilon_s$ is either $1$ or $i$, following: 
\[
  \epsilon_s := \left\{ 
    \begin{array}{cl}
      1 & \text{if $s \equiv 1 \mod 4$},\\
      i & \text{if $s \equiv 3 \mod 4$.}\\
    \end{array}\right.,
\]
for an odd integer $s$.
\end{theorem}

It is possible that the pre-metric groups $(G,q)$ we encounter in the main body of the paper could be properly \emph{pre}-metric---meaning $q$ is degenerate---and the astute reader will notice that Theorem \ref{thm:metricgroupdec} only applies to non-degenerate $q$.  We quickly explain how we may reduce to the non-degenerate case; for details of these constructions, we refer the reader to \cite{Taylor:Gauss}.

If $q$ is degenerate, let $\rad{G}{q}$ denote the \emph{radical}, that is,
\[ \rad{G}{q} = \{ g \in G \mid (\partial q)(g,x) = 0, \ \ \forall x \in G\}. \]
It is a direct computation to check that $q|_{\rad{G}{q}}: \rad{G}{q} \to \Q/\Z$ is always a group homomorphism.  We say the degenerate form $q$ is \emph{tame} if $q|_{\rad{G}{q}}$ is the trivial homomorphism; otherwise, $q$ is \emph{wild}.  Given access to a form, tameness and wildness can be readily decided (recall that we may assume we know the primary decomposition of $G$): simply compute a set of generators of $\rad{G}{q}$ using linear algebra on the Gram matrix of $\partial q$, and then test if $q$ vanishes on all the generators; if so, then $q$ is tame; if not, then $q$ is wild.

For tame degenerate forms, there is a well-defined quotient quadratic form $\overline{q}: G/\rad{G}{q} \to \Q/\Z$, which is \emph{always} non-degenerate, and hence $(G/\rad{G}{q}, \overline{q})$ is a metric group.  Now it is not necessarily true that $(G,q)$ splits as an orthogonal sum $(G/\rad{G}{q}, \overline{q}) \oplus (\rad{G}{q},0)$, but it is true that
\[ \Theta(G,q) = \sqrt{|\rad{G}{q}|}\Theta(G/\rad{G}{q}, \hat{q}).\]
Thus, we may compute the Gauss sum $\Theta(G,q)$ of a degenerate form if it is tame.

For wild degenerate forms, there is no well-defined notion of quotient form.  Fortunately, for our purposes, all we are after are Gauss sums, and it is simple to argue that the Gauss sums of wild degenerate forms always vanish:
\[ \Theta(G,q) =0.\]

Therefore, we may focus our analysis on metric groups.~\cite{BasakJ15} give a constructive proof of Theorem~\ref{thm:metricgroupdec} with an algorithm to normalize the Gram matrix of the bilinear pairing $\partial q$ associated to a metric group $(G,q)$ in order to compute the orthogonal decomposition of $(G,q)$ into irreducible metric groups $(G_i,q_i)_{i=1, \ldots, k}$. They then give the correspondence between irreducible metric groups and their Gauss sum, allowing us to recover the Gauss sum of $(G,q)$ by taking the product of the Gauss sums of the irreducible metric groups $(G_i,q_i)_{i=1, \ldots, k}$ of the decomposition (Lemma~\ref{lem:gaussusummult}).

\medskip

The point of the section is to justify the strictly cubic complexity analysis of the algorithm of~\cite{BasakJ15} as highlighted in Theorem~\ref{thm:evaluategausssum}, by controlling the potential blow-up of matrix coefficients (in $\Q/\Z$) when normalizing the Gram matrix. We recall~\cite{BasakJ15}'s result:

\begin{theorem}[\cite{BasakJ15}]
\label{thm:basakjalgo}
Let $(G,q)$ be a metric group, where $G$ is a $p$-group for a prime $p$. Given the $(d \times d)$-Gram matrix of $\partial q$ for a family of generators $x_1, \ldots , x_d$ of $G$, there is an algorithm to compute the orthogonal decomposition of $(G,q) = (G_1,q_1) \perp \ldots \perp (G_k,q_k)$ into irreducible metric groups, and to deduce $\Theta(G,q)$.
\end{theorem}

We prove the complexity part of:

\begin{lemma}
\label{lem:cpxnormalizegrammatrix}
For a {\em fixed} metric group $(G,q)$, where $G$ is a $p$-group, the algorithm of Theorem~\ref{thm:basakjalgo} terminates in $O(d^3 \poly(\log |G|))$ operations, using $O(d^2 \poly(\log |G|))$ memory words.
\end{lemma}

\begin{proof}
We refer to~\cite{BasakJ15} for the details of the algorithm, and only highlight here the key operations. The computation of the decomposition is done by normalizing the Gram matrix into diagonal form (for odd prime $p$), and block diagonal form with blocks of size $1$ and $2$ (for $p=2$). Each block of the diagonal represents an irreducible form as enumerated in the table of the classification Theorem~\ref{thm:metricgroupdec}. 

The algorithm is a cubic reduction into normal form via Gaussian elimination, using arithmetic operations in $\Q/\Z$. In consequence, we only need to prove that the complexity of all arithmetic operations, and the memory size of all matrix coefficients (in $\Q/\Z$), remain bounded by a polynomial function in $\log |G|$ (and in particular is independent of the dimension of the matrix) during the reduction. 

The group $G$ being a $p$-group for a prime $p$, denote by $p^r$ the maximal order of any element in the set of generators. The algorithm of~\cite{BasakJ15} proceeds by exchanges of rows and columns, and cancellations of matrix entries by elementary row/column operations of the form $\row_i \leftarrow \row_i - u \row_j$ and $\col_i \leftarrow \col_i - u \col_j$ for some integer $u, |u| < p^r$. 

Before normalization, all entries of the Gram matrix are fractions of the form $\frac{x}{p^r} \in \Q/\Z$, for $0 \leq x < p^r$. Indeed, by bilinearity of $\partial q$, we have $\partial q( p^r x_i,x_j) = \partial q(0,x_j) = 0 = p^r \partial q(x_i,x_j)$, hence $p^r \partial q(x_i,x_j) \in \Z$ for any $x_i,x_j$. Considering the matrix updates above, all entries remain of the form $\frac{x}{p^r}$ during the reduction. In particular, the maximal order of any entry of the matrix is always a power of $p$ smaller or equal to $p^r$.

The matrix reduction also requires the following arithmetic operations: for $0 \leq a < p^r$ and $0 < k \leq r$, basic arithmetic operations ($+,-,\times$) in $\Z/p^k\Z$, inversion in $\Z/p^k\Z$, solving in $x$ the quadratic residue equation $x^2 \equiv a (\mod p^k)$, greatest power of $p$ dividing $a$. All these operations can be solved in deterministic $O(\poly(r \log p) )$ operations for a small polynomial $\poly(\cdot)$---except solving quadratic residues---using standard arithmetic operations (notably the (extended) Euclidean algorithm for gcd and inverses modulo an integer). Solving $x^2 \equiv a (\mod p^k)$ can be done with a Las Vegas algorithm in (expected) $O(\poly(r \log p) )$ operations for a small polynomial $\poly(\cdot)$ (using a blend of techniques such as Tonelli-Shanks algorithm~\cite{Shanks73} and Hensel's lifting lemma~\cite[Thm~7.3]{Eisenbud95} in modular arithmetic). 

The complexity $O(d^3 \poly(\log |G|))$ follows from the fact than $p^k \leq |G|$. The bound $O(d^2 \poly( \log |G|))$ for memory complexity is the size required to store the $(d\times d)$ Gram matrix, where matrix entries in $\Q/\Z$ are represented by $\poly(\log |G|)$ bits (by, \eg, pairs of integers (nominator, denominator) ).
\end{proof}

Finally, we can prove Theorem~\ref{thm:evaluategausssum}:

\begin{proof}[Proof of Theorem~\ref{thm:evaluategausssum}]
Let $(G,q)$ be a pre-metric group, given by its primary decomposition with $d=d(G)$ summands, and an oracle access $x \mapsto q(x)$ for any $x \in G$. We treat each $p$-subgroups of $G$ separately (by virtue of Lemma~\ref{lem:orthopq} and Lemma~\ref{lem:gaussusummult}) and return the product of the Gauss sums of the $p$-subgroups, for all primes $p$. Each Gauss sum of the $p$-subgroups is computed with Theorem~\ref{thm:basakjalgo} and the complexity bounds of Theorem~\ref{lem:cpxnormalizegrammatrix}. The complexity of Theorem~\ref{lem:cpxnormalizegrammatrix} being convex super-polynomial functions, the worst case complexity of the algorithm is attained for a single $p$-group $G$ with $|G|$ elements and $d(G)$ summands in its primary decomposition.
\end{proof}

\section{Primer on TVBW quantum invariants from spherical fusion categories}
\label{sec:quantumtopology}
Fusion categories are algebraic structures that abstract the properties needed to organize the representation theory of finite groups and more general objects (such as finite-dimensional semi-simple Hopf algebras~\cite{EGNO}). Fusion categories that are {\it spherical} give rise to 3-manifold invariants via {\it topological quantum field theory} (TQFT), a functorial assignment of vector spaces and linear maps to topological manifolds and cobordisms, respectively.  Such 3D TQFTs are called \emph{Turaev-Viro} or \emph{Turaev-Viro-Barrett-Westbury (TVBW)} TQFTs.  The invariants of closed triangulated 3-manifolds they provide can be computed from their associated fusion category via a {\it state sum construction}~\cite{TuraevViro,BarrettWestbury}.   We review the details of this construction here.

\subsection{The data type of a spherical fusion category}
\label{sec:fusion}
This section contains the equations that define a {\it skeletal spherical fusion category} and explains why their solutions give a finite data type for a spherical fusion category.  We follow the combinatorial approach to defining spherical fusion categories along the lines of \cite{Bonderson:thesis} and, especially \cite{WangCBMS}, since this point-of-view is required for our algorithmic analysis; for a fully general categorical definition see \cite{EGNO}. 

Every spherical fusion category $\calC$ admits a combinatorial description in which the objects and their defining coherence axioms and properties can be encoded as the sets of solutions to a finite collection of polynomial equations. Given a finite label set $L=\{1,a,b,c,\ldots\}$ that indexes a set of representatives of isomorphism classes of {\it simple objects} of $\calC$, the data of a spherical fusion category $\mathcal{C}$ consists of
\begin{enumerate}
\item the non-negative integer {\it fusion coefficients} $N^{ab}_c$,
\item the {\it $F$-symbols} $[F^{abc}_d]_{e,f} \in \C$, and 
\item the {\it pivotal coefficients} $t_a \in \{\pm 1\}$ 
\end{enumerate}
satisfying the equations defining a {\it fusion ring}, the {\it triangle} and {\it pentagon equations}, and the {\it spherical pivotal equations}, respectively.\footnote{Since we are eventually interested in Tambara-Yamagami categories we will only need to consider {\it multiplicity-free} spherical fusion categories,~\ie,~those with $N^{ab}_c \in \{0,1\}$. The $F$-symbols $[F^{abc}_d]_{(e,\alpha,\beta),(f,\delta,\gamma)}$ for theories with multiplicity require additional parameters $\alpha,\beta,\delta,\gamma$.} One should think of such a collection $\{N^{ab}_c, [F^{abc}_d]_{e,f}, t_a : a,b,c,d,e,f \in L\}$---called a {\it skeletal} spherical fusion category---as a finite data type instantiating a spherical fusion category $\mathcal{C}$.  (We note that the $F$-symbols can always be taken to be algebraic numbers over the rationals, and thus can be described with a finite amount of data.)

Let $L=\{1,a,b,c, \ldots\}$ be a finite set equipped with an involution $*: L \to L$ called {\it duality} that fixes the distinguished element $1 \in L$. 

\paragraph{The data type of a fusion ring}
The free $\Z$-module $F$ with {\it basis} $L$ whose elements are finite integer-linear combinations of elements of $L$ has the structure of a {\it fusion ring} if $F$ admits a unital, associative multiplication such that (1) every product of basis elements is a non-negative integer combination of basis elements, \ie,
\[ a \times b = \sum_{c \in L} N^{ab}_c\, c \qquad N^{ab}_c \in \Z_{\ge 0}\]
and (2) the involution on basis elements extends to an anti-involution on $F$, so that $(a\times b)^* = b^* \times a^*$. The number $|L|$ is called the {\it rank} of the fusion ring.  We note that the multiplication in the ring need not be commutative; equivalently, we do not require $N^{ab}_c = N^{ba}_c$.

Data satisfying the axioms of a fusion ring $F=(L,+,\times)$ can be encoded by the {\it fusion coefficients} $N^{ab}_c \in \Z_{\ge 0}$ satisfying the equations
\begin{enumerate} 
\item $\sum_x N^{ab}_xN^{xc}_d = \sum_x N^{bc}_xN^{ax}_d$ \hfill (associativity)
\item $N^{1a}_b = N^{a1}_b = \delta_{ab}$ \hfill  (unitality) 
\item $N^{a^*b}_1 = N^{ba^*}_1 = \delta_{ab}$ \hfill (duality)
\item $N^{ab}_c = N^{a^*c}_b = N^{cb^*}_a$ \hfill (Frobenius reciprocity)
\end{enumerate}
where $\delta_{ab}$ is the Kronecker delta function. We will only need to consider fusion rings which are {\it multiplicity-free}, \ie, those with $N^{ab}_c \in \{0,1\}$ for all $a,b,c \in L$. Multiplicity-free fusion rings are classified through at least rank $6$~\cite{Palcoux2022}. 

\paragraph{The data type of a fusion category}
Fusion {\it categories} can be defined over any algebraically closed field $\mathbb{k}$, although in most applications to topology, physics, and computer science one is interested in fusion categories defined over $\mathbb{C}$. The data of a (multiplicity-free) {\it skeletal fusion category} over $\mathbb{C}$ is the collection of numbers $\{ N^{ab}_c, [F^{abc}_d]_{e,f} : a,b,c,d,e,f \in L \}$  consisting of the label set $L$, fusion ring coefficients $N^{ab}_c \in \{0,1\}$ and {\it $F$-symbols} $[F^{abc}_d]_{e,f} \in \mathbb{k}$.

The $F$-symbols are equal to $0$ whenever $N^{ab}_e\,N^{ec}_d=0$ or $N^{bc}_fN^{af}_d=0$, otherwise the { \it $F$-matrices} $[F^{abc}_d]$ must be invertible. The $F$-symbols must further satisfy the {\it triangle equations}: 
\[ [F^{1bc}_d]_{b,d} = 1 \qquad [F^{a1c}_d]_{a,c} = 1 \qquad [F^{ab1}_d]_{d,b}=1\]
and the {\it pentagon equations}:
\[ \sum_h [F^{abc}_g]_{f,h}[F^{ahd}_e]_{g,k}[F^{bcd}_k]_{h,l} = [F^{fcd}_e]_{g,l}[F^{abl}_e]_{f,k} \]
for all $a,b,c,d \ldots \in L$. 
The $F$-symbols also satisfy a property called {\it rigidity} so that $[F^{aa^*a}_a]_{1,1}= [F^{a^*aa^*}_{a^*}]^{-1}_{1,1}$.

Notice that these are integral polynomial equations, and hence the $F$-symbols may be encoded as algebraic numbers over $\mathbb{Q}$. The details of these encodings are inessential for our purposes, since we will be working with a fixed fusion category. We make this remark only to point out that the data type of a skeletal fusion category over $\mathbb{C}$ is in fact finite. 

Finding solutions to the pentagon equations (and thus proving the existence of a fusion category) with the aid of a computer is an important frontier in the program to classify fusion categories, see for example~\cite{Wolf}.

Two sets of $F$-symbols $[F^{abc}_d]_{e,f}$ and $[G^{abc}_d]_{e,f}$ represent the same multiplicity-free fusion category if there exists $\gamma^{ab}_c \in \mathbb{k}$ (sometimes called a {\it gauge transformation}) which are 0 whenever $N^{ab}_c=0$ and otherwise must satisfy: 
\[\gamma^{1a}_a=1 \qquad \gamma^{a1}_a = 1 \]
and:
\[ [G^{abc}_d]_{e,f} = \gamma^{ab}_e\,\gamma^{ec}_d\,[F^{abc}_d]_{e,f} (\gamma^{bc}_f)^{-1}(\gamma^{af}_d)^{-1} \]
for all $a,b,c,d,e,f \in L$. 

\paragraph{The data type of a spherical fusion category} Apart from the rank $|L|$ and the Frobenius-Perron eigenvalues of the {\it fusion matrices} $(N_a)_{bc}:= N^{ab}_c$ there are few meaningful {\it gauge invariant} quantities of fusion categories one can extract from the data $\{N^{ab}_c, [F^{abc}_d]_{e,f} \}$. Gauge invariants of fusion categories become more readily available upon the addition of a {\it pivotal} structure, which can be encoded as $t_a \in \mathbb{k}$ satisfying the {\it pivotal equations}

\begin{enumerate}
\item $t_1=1$
\item $t_a = t_{a^*}$
\item $\frac{t_c}{t_at_b}= [F^{abc^*}_1]_{c,a^*}[F^{bc^*a}_1]_{a^*b^*}[F^{c^*ab}_1]_{b^*c} = [F^{abc^*}_1]^{-1}_{c,a^*}[F^{bc^*a}_1]^{-1}_{a^*b^*}[F^{c^*ab}_1]^{-1}_{b^*c}$.
\end{enumerate}
When $t_a \in \{\pm 1\}$ then we say the corresponding fusion category is {\it spherical}.

Together the fusion coefficients, $F$-symbols, and (spherical) pivotal coefficients \[\{N^{ab}_c, [F^{abc}_d]_{e,f}, t_a : a,b,c,d,e,f \in L\}\] define a spherical fusion category. Conversely, every multiplicity-free spherical fusion category is equivalent to a skeletal fusion category described by such data.

\subsection{The graphical calculus in a skeletal spherical category} 

Associated to the data of a skeletal spherical fusion category introduced in Appendix~\ref{sec:fusion} is a coherent diagrammatic calculus on trivalent graphs embedded in the 2-sphere $S^2$. Our goal here is simply to give a minimal overview of this diagrammatic calculus, focusing only on those aspects which are needed to define the Turaev-Viro-Barrett-Westbury invariant of a triangulated 3-manifold. References with further details include~\cite{WangCBMS} or~\cite{CuiWang:multi}.   

Let $\{N^{ab}_c, [F^{abc}_d]_{e,f}, t_a : a,b,c,d,e,f \in L\}$ be a skeletal spherical fusion category and let $\Gamma$ be a directed trivalent graph in the plane with edges labeled by elements of $L$.  The orientation of an edge labeled by $a \in L$ can be reversed by replacing the label with its {\it dual} $a^* \in L$.\footnote{Some care must be taken with the orientations of edges at critical points of the graph but we will not worry about this here.} Edges labeled by $1 \in L$---which we will indicate by drawing a dashed line and for which the orientation is not important since $1^*=1$---can be erased or inserted at will. 
\[\begin{tikzpicture}[line width=1,baseline=.5cm]]
\draw (0,.5) --(0,0);
\draw[<-](0,.5)--(0,1);
\draw (.25,.5) node[right] {$a$};
\end{tikzpicture}
\quad =\quad
\begin{tikzpicture}[line width=1,baseline=.5cm]]
\draw[->] (0,0) --(0,.5);
\draw(0,.5)--(0,1);
\draw (.25,.5) node[right] {$a^*$};
\end{tikzpicture}
\qquad \qquad
\begin{tikzpicture}[line width=1,baseline=.5cm]]
\draw (0,.5) --(0,0);
\draw[](0,.5)--(0,1);
\draw (.25,.5) node[right] {$1$};
\end{tikzpicture}
\quad =\quad
\begin{tikzpicture}[line width=1,baseline=.5cm]]
\draw[dashed] (0,0) --(0,.5);
\draw[dashed,](0,.5)--(0,1);
\end{tikzpicture}
\quad =\quad
\begin{tikzpicture}[line width=1,baseline=.5cm]]
\draw (,.6) node{$\emptyset$};
\end{tikzpicture}
\]
We say that a trivalent vertex with positively oriented edges labeled by $a,b,c \in L$ is {\it admissibly labeled} if $N^{ab}_c \ne 0$. Admissibility of vertices with other edge orientations are defined analogously using duality.
\[
\begin{tikzpicture}[line width=1,baseline=-.5cm, scale=.75]
\draw[->] (0,0) node[above] {$a$} -- (.25,-.25);
\draw (.25,-.25)--(.5,-.5);
\draw (.5,-.5) -- (.5,-1) node[below] {$c$};
\draw[-<] (.5,-.75);
\draw[->] (1,0) node[above] {$b$} -- (.75,-.25);
\draw (.75,-.25)--(.5,-.5);
\end{tikzpicture}
\quad \text{ and } 
\begin{tikzpicture}[line width=1,baseline=.375cm, scale=-.75]
\draw (0,0) node[below] {$b$} -- (.25,-.25);
\draw[<-] (.25,-.25)--(.5,-.5);
\draw (.5,-.5) -- (.5,-1) node[above] {$c$};
\draw[-<] (.5,-.75);
\draw (1,0) node[below] {$\phantom{b}a\phantom{b}$} -- (.75,-.25);
\draw[<-] (.75,-.25)--(.5,-.5);
\end{tikzpicture}
\quad \text{ admissibly labeled if } N^{ab}_c \ne 0
\]
Then $\Gamma$ is admissibly labeled if all of its vertices are admissibly labeled. Given an admissibly labeled pair of neighboring vertices, one can change the local connectivity by an {\it $F$-move}, at the cost of passing to a formal sum of graphs, each weighted by an appropriate $F$-symbol. 
\[\begin{tikzpicture}[line width=1,baseline=-12.5,scale=.75]
\draw[->] (0,0) node[above] {$a$} -- (.25,-.25);
\draw[->] (.25,-.25)--(.75,-.75);
\draw (.75,-.75)--(1,-1);
\draw[->] (1,-1) -- (1,-1.25);
\draw (1,-1.25) -- (1,-1.5) node[below] {$d$};
\draw[->] (1,0) node[above] {$b$} --(.75,-.25);
\draw (.75,-.25)-- (.5,-.5);
\draw[->] (2,0) node[above] {$c$}-- (1.5,-.5);
\draw (1.5,-.5)-- (1,-1);
\draw (.75,-1) node[left] {$e$};
\end{tikzpicture}
=
\sum_{f} [F^{abc}_d]_{e,f}
\begin{tikzpicture}[line width=1,baseline=-12.5,scale=.75]
\draw[->] (0,0) node[above] {$a$} -- (.5,-.5);
\draw (.5,-.5)-- (1,-1); 
\draw[->] (1,-1) -- (1,-1.25);
\draw (1,-1.25) -- (1,-1.5) node[below] {$d$};
\draw[->] (1,0) node[above] {$b$} -- (1.25,-.25);
\draw (1.25,-.25) -- (1.5,-.5);
\draw[->] (2,0) node[above] {$c$}-- (1.75,-.25);
\draw[->] (1.75,-.25) --(1.25,-.75);
\draw (1.25,-.75)--(1,-1);
\draw (1.25,-1) node[right] {$f$};
\end{tikzpicture}
\qquad \qquad
\begin{tikzpicture}[line width=1,baseline=-12.5,scale=.75]
\draw[->] (0,0) node[above] {$a$} -- (.5,-.5);
\draw (.5,-.5)-- (1,-1); 
\draw[->] (1,-1) -- (1,-1.25);
\draw (1,-1.25) -- (1,-1.5) node[below] {$d$};
\draw[->] (1,0) node[above] {$b$} -- (1.25,-.25);
\draw (1.25,-.25) -- (1.5,-.5);
\draw[->] (2,0) node[above] {$c$}-- (1.75,-.25);
\draw[->] (1.75,-.25) --(1.25,-.75);
\draw (1.25,-.75)--(1,-1);
\draw (1.25,-1) node[right] {$e$};
\end{tikzpicture}
= \sum_{f} [F^{abc}_d]^{-1}_{e,f}
\begin{tikzpicture}[line width=1,baseline=-12.5,scale=.75]
\draw[->] (0,0) node[above] {$a$} -- (.25,-.25);
\draw[->] (.25,-.25)--(.75,-.75);
\draw (.75,-.75)--(1,-1);
\draw[->] (1,-1) -- (1,-1.25);
\draw (1,-1.25) -- (1,-1.5) node[below] {$d$};
\draw[->] (1,0) node[above] {$b$} --(.75,-.25);
\draw (.75,-.25)-- (.5,-.5);
\draw[->] (2,0) node[above] {$c$}-- (1.5,-.5);
\draw (1.5,-.5)-- (1,-1);
\draw (.75,-1) node[left] {$f$};
\end{tikzpicture}\]

Once equipped with a spherical pivotal structure, the value of a closed loop labeled by $a \in L$ becomes an invariant quantity called the {\it quantum dimension} of $a$ and denoted by $d_a$. The quantum dimensions are given in terms of the skeletal data by $d_a = \begin{tikzpicture}[line width=1,baseline=-2.5] \draw (0,0) circle (.25); \draw[-<] (-.25,0) node[left] {$a$}; \end{tikzpicture} = t_a^{-1}[F^{aa^*a}_a]_{11}^{-1}$. 
The category then inherits a {\it global quantum dimension} $D$, defined as:
\[ D = \sum_{a \in L} d_a^2.\] 
The pivotal coefficients provide a way to close up or {\it trace off} diagrams by introducing a factor of $t_a$ and gluing together certain special trivalent vertices called {\it cups} and {\it caps}. Due to the pivotal coefficients satisfying the spherical equations there is no ambiguity and tracing off can be done either to the left or the right, so that diagrams can be interpreted as embedded on the surface of a sphere rather than the plane (hence the name spherical).

Sphericality also endows diagrams with a degree of isotopy invariance. For our purposes it suffices to establish conventions for how to perform the {\it bubble-popping move}.
The spherical structures on Tambara-Yamagami categories that we will study satisfy a condition called \emph{pseudo-unitarity}, which allows us to define the following especially nice rule for bubble-popping (see~\cite{CuiWang:multi}):
\vspace{-.2cm}
\[\begin{tikzpicture}[line width=1,baseline=0, scale=.75]
\draw[->] (0,0) node[left] {$\phantom{b}a\phantom{b}$} -- (.25,-.25);
\draw (.25,-.25)--(.5,-.5);
\draw(.5,-.5) -- (.5,-1) node[right] {$c'$};
\draw[-<] (.5,-.75);
\draw[->] (1,0) node[right] {$b$} -- (.75,-.25);
\draw (.75,-.25)--(.5,-.5);
\begin{scope}[yscale=-1]
\draw (0,0)  -- (.25,-.25);
\draw[<-] (.25,-.25) --(.5,-.5);
\draw (.5,-.5) -- (.5,-1) node[left] {$c$};
\draw[-<] (.5,-.75);
\draw (1,0)  -- (.75,-.25);
\draw[<-] (.75,-.25) --(.5,-.5);
\end{scope}
\end{tikzpicture}
\quad = \quad \delta_{c,c'}\frac{\sqrt{d_ad_b}}{\sqrt{d_c}} \begin{tikzpicture}[line width=1,baseline=0, scale=.5]
\draw (0,-1) node[below] {$c$} -- (0,1) node[above] {$c$};
\draw[-<] (0,0);
\end{tikzpicture}.\]

In summary, the admissibility of labelings of trivalent graphs is dictated by the fusion coefficients, and sphericality gives a well-defined value to closed loops. The various $F$-moves and bubble-popping moves provide local rules that allow us to simplify arbitrary closed trivalent planar graphs until we are left with complex numbers, which are necessarily isotopy invariants of the initial labeled planar graphs.

\subsection{TVBW state sum invariants of 3-manifolds}
Next we explain how the graphical calculus associated to a multiplicity-free skeletal spherical fusion category $\mathcal{C}$ is used to define the Turaev-Viro-Barrett-Westbury state sum invariant of a (closed and oriented) triangulated 3-manifold. Let $\tri$ be a triangulation of a closed 3-manifold $M$, with vertices $V$, edges $E$, faces $F$, and tetrahedra $T$. Recall from Section~\ref{ss:top} that we fix an arbitrary orientation of the edges of $\tri$. A {\em coloring} $\theta: E \to L$ of $\tri$ is an assignment of an element of the label set $L$ to each edge of the triangulation.

We now define the weights involved in the definition of the Turaev-Viro-Barrett-Westbury state sum invariant for a given coloring $\theta$.  Let $\theta_e:=\theta(e) \in L$ denote the label of an edge $e \in E$.  
\begin{enumerate}
\item To each vertex $v \in V$ we shall assign the global quantum dimension $D$.
\item To each edge $e \in E$ we assign the quantum dimension $d_{\theta_e}$.
\item To each face $f \in F$ with edges $e_1,e_2,e_3$ we assign the factor $\sqrt{d_{\theta_{e_1}}d_{\theta_{e_2}}d_{\theta_{e_3}}}$.
\item To each tetrahedron $t \in T$ with edges $e_1,e_2, \ldots e_6$ and orientation $O(t) = \pm$, we assign a {\it tetrahedral weight} $|t|^{O(t)}_{\theta}$, whose definition follows.
\end{enumerate}

Let $(\Delta,+)$ and $(\Delta,-)$ be model oriented tetrahedra with edges $a,b,c,d,e,f$, equipped with model edge orientations as follows:
\vspace{-.2cm}
\[
\begin{tikzpicture}[scale=.45,baseline=20,decoration={
    markings,
    mark=at position 0.5 with {\arrow{>}}}]
\draw (3,-3) node{$(\Delta,+)$};
\draw[thick,postaction={decorate}] (0,0)--(4,-1.5) node[midway,below] {$b$};
\draw[thick,postaction={decorate}] (0,0)--(3,5) node[midway,left] {$d$};
\draw[thick,postaction={decorate}] (3,5)--(4,-1.5) node[midway,right] {$f$};
\draw[thick,postaction={decorate}] (6,0)--(4,-1.5) node[midway,below] {$a$};
\draw[thick,postaction={decorate}] (6,0)--(3,5) node[midway,right] {$e$};
\draw[dashed,thick,postaction={decorate}] (0,0)--(6,0) node[midway,above] {$c$};
\begin{scope}[xshift=11cm]
\draw (3,-3) node{$(\Delta,-)$};
\draw[thick,postaction={decorate}] (0,0)--(4,-1.5) node[midway,below] {$b$};
\draw[thick,postaction={decorate}] (0,0)--(3,5) node[midway,left] {$d$};
\draw[thick,postaction={decorate}] (3,5)--(4,-1.5) node[midway,right] {$f$};
\draw[thick,postaction={decorate}] (6,0)--(4,-1.5) node[midway,below] {$a$};
\draw[thick,postaction={decorate}] (6,0)--(3,5) node[midway,right] {$e$};
\draw[dashed,thick,postaction={decorate}] (0,0)--(6,0) node[midway,above] {$c$};
\end{scope}
\end{tikzpicture}
\]
(As usual, the edge orientations in these models are arbitrary, and have nothing to do with the tetrahedron's orientation $\pm$.)  If $\theta$ is a labeling of the edges of $\Delta$ by $\theta_a,\theta_b,\theta_c,\theta_d,\theta_e,\theta_f$, then we define the \emph{model tetrahedral weights}, denoted equivalently either by $|\Delta|_\theta^\pm$ or by $\sixj{\theta_{a}}{\theta_{b}}{\theta_{c}}{\theta_{d}}{\theta_{e}}{\theta_{f}}^\pm$, via the graphical calculus:
\[\begin{aligned}
|\Delta|^{+}_{\theta} = 
\sixj{\theta_{a}}{\theta_{b}}{\theta_{c}}{\theta_{d}}{\theta_{e}}{\theta_{f}}^+ &:=
\begin{tikzpicture}[scale=.4,baseline=-12]
\tetsymbleft{a}{b}{c}{d}{e}{f}
\end{tikzpicture} = [F^{abd}_e]_{c,f}\sqrt{d_ad_bd_dd_e} \\
|\Delta|^{-}_{\theta} = 
\sixj{\theta_{a}}{\theta_{b}}{\theta_{c}}{\theta_{d}}{\theta_{e}}{\theta_{f}}^- &:=
\begin{tikzpicture}[scale=.4,baseline=-12]
\tetsymbright{a}{b}{c}{d}{e}{f}
\end{tikzpicture} = [F^{abd}_e]^{-1}_{f,c}\sqrt{d_ad_bd_dd_e}
\end{aligned}\]
We call a coloring of the edges of these model tetrahedra \emph{admissible} if the induced labeled trivalent graphs seen here are admissibly labeled.

For each $t \in T$ with orientation $O(t)$, we fix an orientation preserving isomorphism $I_t: t \to (\Delta,O(t))$.  If $\theta$ is our given coloring of $\tri$, then we use the coloring of the edges of $t$ to induce a coloring of $\Delta$; however, as the edge orientations of $t$ and $\Delta$ were chosen arbitrarily, there is no promise that $I$ matches the edge orientations, and any time the orientation of an edge $e$ of $t$ does not match the orientation of the edge $I_t(e)$ of $\Delta$, we must color $I_t(e)$ by $\theta_e^*$ instead of $\theta_e$.  In other words, we define the push-forward coloring to $(\Delta,O(t))$ along $I_t$ by the formula:
\[(I_t \theta)(e) :=
\begin{cases}
\theta_{I^{-1}(e)} & \text{ if orientations of $e$ and $I^{-1}(e)$ match}\\
\theta_{I^{-1}(e)}^* & \text{ else}.
\end{cases}
\]
Finally, define the tetrahedral weight:
\[ |t|_\theta^{O(t)} := |\Delta|_{I_*\theta}^{O(t)}.\]

We say a coloring $\theta: E \to L$ of $\tri$ is {\em admissible} if for every oriented tetrahedron $t \in T$, the coloring $I_t(\theta)$ of the model tetrahedron is admissible.  (The tetrahedral weight of an inadmissible coloring is always 0.)  Denote the set of admissible colorings of $\tri$ by $\adm_L(\tri)$. Then the {\it state sum} of $\M$ with triangulation $\tri$ is given by:
\[|M|_{\mathcal{C},\tri, \{I_t\}_{t \in T}} = 
  \displaystyle\sum_{\theta \in \adm_L(\tri)} 
    \frac{ 
      \displaystyle\prod_{t \in T} 
      |t|_\theta^{O(t)}
      \prod_{e \in E} d_{\theta{e}} 
      }
      {
      \displaystyle\prod_{f \in F} \sqrt{d_{\theta{e_1}}d_{\theta{e_2}}d_{\theta{e_3}}} \prod_{v \in V}D.
      }
\]
Remarkably, the structures of a spherical fusion category guarantee that $|M|_{\mathcal{C},\tri, \{I_t\}_{t \in T}}$ does not depend on the choices of isomorphisms $I_t$ for the tetrahedra $t \in T$, and so we may suppress it from our notation (although going forward, we will include a choice of model isomorphisms $\{I_t\}_{t \in T}$ in the data structure of our triangulations, as these are necessary to compute the invariant).  Even more is true.

\begin{theorem}[Topological invariance and multiplicativity of state sum~\cite{BarrettWestbury, Turaev:book, TuraevVirelizier:approaches}]
\label{thm:toppropty}
Let $\mathcal{C}$ be a spherical fusion category with global quantum dimension $D = \sum_{a \in L} d_a^2$. The state sum defined from $\mathcal{C}$ is a {\em topological invariant}, \ie, for any closed 3-manifold $\M$ and $\tri_1, \tri_2$ generalized triangulations of $\M$, we have:
\[
  |\M|_{\mathcal{C},\tri_1} = |\M|_{\mathcal{C},\tri_2}
\]
and we write $|\M|_{\mathcal{C}}$. Additionally, if $\M = \M_1 \# \M_2$, then:
\[
  |\M|_{\mathcal{C}} = D|\M_1|_{\mathcal{C}}|\M_2|_{\mathcal{C}}
\]
\end{theorem}

\subsection{Tambara-Yamagami category and its state sum invariant} 
A {\em Tambara-Yamagami category}~\cite{TambaraY98} over $\C$, denoted by $\ty(A,\chi,\nu)$, is determined by the data of a finite abelian group $A$ with $|A|$ elements, a non-degenerate symmetric bicharacter $\chi: A \times A \to U(1)$ (\ie, a non-degenerate symmetric bilinear pairing into the group $(U(1),\times)$), and a choice of square root of $|A|^{-1}$ denoted by $\nu \in \{\pm |A|^{-1/2}\}$. Skeletal data for these categories in the language of the previous section are as follows.

The label set of $\ty(A,\chi,\nu)$ is given by $L = A \sqcup \{m\}$, \ie, the elements of the group $A$ together with an extra element we call $m$. The duality involution on $L$ sends $a \mapsto -a$ and $m \mapsto m$. The fusion coefficients involving $a,b,c \in A$ are induced by the group addition of $A$:
\[
N_c^{ab} = 
\begin{cases}
1 & \text{ if } c = a+b \\
0 & \text{ else.}
\end{cases}
\]
Note in particular that $N^{ab}_c=N^{ba}_c$, $N^{1a}_a=1$, and $N^{aa^{-1}}_1=1$. The full set of nonzero fusion coefficients is given by $N^{ab}_{a+b} =1, N^{am}_m=1, N^{ma}_m =1,N^{mm}_a = 1$ for $a,b \in A$ corresponding to the following complete set of admissibly labeled trivalent vertices (up to inverting edge orientations while simultaneously replacing the label with its dual):
\[
\begin{tikzpicture}[line width=1,baseline=-.5cm, scale=.75]
\draw[->] (0,0) node[above] {$a$} -- (.25,-.25);
\draw (.25,-.25)--(.5,-.5);
\draw (.5,-.5) -- (.5,-1) node[below] {$a+b$};
\draw[-<] (.5,-.75);
\draw[->] (1,0) node[above] {$b$} -- (.75,-.25);
\draw (.75,-.25)--(.5,-.5);
\end{tikzpicture} \qquad \begin{tikzpicture}[line width=1,baseline=-.5cm, scale=.75]
\draw[->] (0,0) node[above] {$a$} -- (.25,-.25);
\draw (.25,-.25)--(.5,-.5);
\draw (.5,-.5) -- (.5,-1) node[below] {$m$};
\draw[-<] (.5,-.75);
\draw[->] (1,0) node[above] {$m$} -- (.75,-.25);
\draw (.75,-.25)--(.5,-.5);
\end{tikzpicture} \qquad \begin{tikzpicture}[line width=1,baseline=-.5cm, scale=.75]
\draw[->] (0,0) node[above] {$m$} -- (.25,-.25);
\draw (.25,-.25)--(.5,-.5);
\draw (.5,-.5) -- (.5,-1) node[below] {$b$};
\draw[-<] (.5,-.75);
\draw[->] (1,0) node[above] {$m$} -- (.75,-.25);
\draw (.75,-.25)--(.5,-.5);
\end{tikzpicture}
\qquad \begin{tikzpicture}[line width=1,baseline=-.5cm, scale=.75]
\draw[->] (0,0) node[above] {$m$} -- (.25,-.25);
\draw (.25,-.25)--(.5,-.5);
\draw (.5,-.5) -- (.5,-1) node[below] {$m$};
\draw[-<] (.5,-.75);
\draw[->] (1,0) node[above] {$a$} -- (.75,-.25);
\draw (.75,-.25)--(.5,-.5);
\end{tikzpicture}
\]
The $F$-symbols are: 

\[
\begin{aligned}
  {[F^{abc}_{a+b+c}]}_{a+b,b+c} = 1 \quad &\quad  [F^{amm}_b]_{m,c} = \delta_{b,a+c} \\ 
  [F^{abm}_{m}]_{a+b,m} = 1  \quad & \quad [F^{mmc}_{a}]_{b,m} = \delta_{a,b+c}\\
  [F^{mbc}_{m}]_{m,b+c} = 1 \quad & \quad [F^{mbm}_a]_{m,m} = \chi(a,b)\\  
  [F^{amc}_m]_{m,m}=\chi(a,c) \quad & \quad [F^{mmm}_m]_{a,b} = \nu\overline{\chi(a,b)}\\
\end{aligned}
\]
for $\delta_{a,b} = 1$ if $a=b$, and $0$ otherwise. Following~\cite{TuraevVainerman} we work with the pivotal coefficients
\[\begin{aligned} t_a &=1 \text{ for all $a \in A$, and } \\ t_m &= \nu/|\nu| = \textrm{sign}(\nu). &\end{aligned}\]
This is the unique choice of pivotal structure that results in positive quantum dimensions of all the labels, and thus allows us to employ the bubble-popping conventions described above.  With this pivotal structure, the quantum dimensions are:
\[\begin{aligned}
d_a &=1 \text{ for all $a \in A$, and} \\
d_m &= +|A|^{1/2}.
\end{aligned}\]

It is not in general true that the tetrahedron weights of a spherical fusion category must admit full tetrahedral symmetry.  However, it is close to being true in the case of Tambara-Yamagami categories, as we only have four types of admissible tetrahedron weights for each orientation 
(one can check that all others are equal to one of these):
\[
\begin{array}{cccc}
\text{{\bf $m$-empty}} &&& \text{{\bf $m$-triangle}}\\
\sixj{a}{b}{a+b}{b-c}{a+c}{c}^+ = 1 &&& \sixj{m}{m}{a}{a}{a+b}{m}^+ = d_m\\
\\
\sixj{a}{b}{a+b}{c-b}{a+c}{c}^- = 1 &&& \sixj{a}{b}{a+b}{m}{m}{m}^- = d_m \\
\\
\text{{\bf flat $m$-quad}} &&& \text{{\bf crooked $m$-quad}}\\
\sixj{a}{m}{m}{b}{m}{m}^+ = d_m \chi(a,b)  &&& \sixj{m}{m}{a}{m}{m}{b}^+ = \nu d_m^2 \overline{\chi(a,b)} \\
\\
\sixj{a}{m}{m}{b}{m}{m}^- = d_m \overline{\chi(a,b)}&&& \sixj{m}{m}{a}{m}{m}{b}^- = \nu^{-1} d_m^2\chi(a,b),\\
\end{array}
\]
where $a,b,c \in A$.  

\section{On \#P-hardness of Tambara-Yamagami invariants}
\label{sec:hard}
In this section we prove:
\begin{theorem}
	If $G = \Ztwo$, $\chi(a,b) = \exp(\pi i ab])$ and $\nu = -1/\sqrt{2}$, then the problem of (exactly) computing $|\M|_{\ty(G,\chi,\nu)}$ is $\#\mathsf{P}$-hard.
\label{th:hard}
\end{theorem}
This is essentially the content of~\cite[Cor.~3]{BurtonMS18}, and our proof is essentially the same, except we fill a minor gap.  The proof is a simple application of a prior hardness result of Kirby and Melvin~\cite{KirbyM04}, but requires us to briefly introduce an alternative framework for defining quantum invariants based on surgery presentations of 3-manifolds rather than triangulations.  The gap in~\cite{BurtonMS18} we fill is precisely related to the distinction between encoding via surgery diagrams and encoding via triangulations.

\paragraph{Reshetikhin-Turaev surgery invariants of 3-manifolds}
While the usual way of encoding combinatorial descriptions of 3-manifolds is via triangulations, an alternative is to use surgery diagrams.  Recall that a \emph{surgery diagram} is a diagram of a framed link $L$ embedded inside $S^3$; the framing of the link specifies a Dehn surgery that can be performed on $S^3$ to yield a closed 3-manifold.  The Lickorish-Wallace theorem says that every oriented 3-manifold can be realized by some surgery diagram, and the Kirby calculus is a well-known combinatorial method of understanding when two such surgery diagrams represent homeomorphic 3-manifolds via local moves; see~\cite{Rolfsen} for a textbook introduction to these ideas.  Diagrams of framed links in $S^3$ can be described combinatorially by planar drawings with framing coefficients, and so we might treat surgery diagrams as a finite data type representing 3-manifolds.

Analogous to how TVBW state-sum invariants $|\M|_\calC$ of 3-manifolds are defined by building a tensor network using the data of a spherical fusion category $\calC$ and a triangulation of $M$, one may likewise define quantum invariants $\tau_\calB(\M)$ of 3-manifolds by building a tensor network using the data of a \emph{modular fusion category} $\calB$ and a surgery diagram of $M$.  This is called the Reshetikhin-Turaev construction~\cite{ReshetikhinTuraev:invariants,Turaev:modular}.

For our purposes, we do not need to understand the definition of modular fusion category, and so we refer the reader to any of the standard references such as~\cite{Turaev:book} or~\cite{BakalovKirillov:book} for details.  We simply need two facts:
\begin{enumerate}
\item If $\calB$ is a modular fusion category, then it is also a spherical fusion category in a canonical way.  Thus, modular fusion categories may be used for \emph{both} the Reshetikhin-Turaev construction and the Turaev-Viro-Barrett-Westbury construction of quantum invariants.
\item If $\calB$ is a modular fusion category, then the TVBW state-sum invariant and Reshetikhin-Turaev surgery invariant based on $\calB$ are related as follows:
\[ |\M|_\calB = \tau_{\calB}(M)\overline{\tau_{\calB}(M)},\]
where as usual $\overline{\tau_{\calB}(M)}$ is the complex conjugate of $\tau_{\calB}(M)$.
\end{enumerate}
The first fact follows immediately from the definition of modular fusion category, and the second fact is proved by various authors (at various levels of generality), including~\cite{TuraevVirelizier:approaches,Balsam,TuraevVirelizier:homotopyapproaches}.

\paragraph{Proof of Theorem \ref{th:hard}}
In~\cite{KirbyM04}, Kirby and Melvin work with a certain modular fusion category $\calB$ such that computing the Reshetikhin-Turaev invariant of a 3-manifold $M$---\emph{when encoded via a surgery diagram $L$}---is \#P-hard.  More precisely, given a Boolean cubic function $f$ in $n$ variables of the form
\[ f(x_1,\dots,x_n) = \sum_{i=1}^n c_i x_i + \sum_{1\le i < j \le n} c_{ij} x_ix_j + \sum_{1 \le i <j < k \le n}^n c_{ijk} x_ix_jx_k\]
where $c_i,c_{ij},c_{ijk} \in \Ztwo$, they build a surgery diagram $L_f$ with the property that
\[ \tau_\calB(M_f) = 2\times \#\{(x_1,\dots,x_n) \in (\Ztwo)^n \mid f(x_1,\dots,x_n) = 0 \mod 2\} - 2^n\]
where $M_f$ is the 3-manifold encoded by $L_f$.  They then show that counting solutions to $f(x_1,\dots,x_n) = 0 \mod 2$ is \#P-hard (via non-parsimonious reduction), and thereby conclude the same for $ \tau_\calB(M_f)$, at least when $M_f$ is presented by a surgery diagram.

The surgery diagram $L_f$ constructed by Kirby and Melvin is quite simple: start with an unlink with $n$ components labeled by the variables $x_1,\dots,x_n$.  For each $i$ with $c_i=1$, tie a trefoil knot in the corresponding component; for each $i,j$ with $c_{ij}=1$, arrange so that the corresponding components are linked like a Whitehead link; for each $i,j,k$ with $c_{ijk}=1$, link the components like the Borromean rings.  Now take the 0-framing to get $L_f$.

The proof of~\cite[Cor.~3]{BurtonMS18} asserts that Kirby and Melvin furthermore provide a \emph{triangulation} of the manifold $M_f$ encoded by the 0-framed surgery $L_f$, but this requires further arguments, and so we address this here.  Fortunately, it is folklore, using standard methods, that, given any link diagram $L$, the 0-framed surgery $M$ of $L$ can be triangulated in time polynomial in the crossing number of $L$. We sketch the basic idea here. 

Begin by triangulating the complement $S^3 \setminus N(L)^\circ$, where $N(L)^\circ$ is an open regular neighborhood of $L$; see~\cite{HassLagariasPippenger} or~\cite[Rmk.~6.2]{MariaP19} for some standard methods of doing so.  Critically, since we specifically wish to build the \emph{0-framed} surgery manifold, homological calculations enable us to efficiently draw the surgery slopes on each of the torus boundary components of our triangulation using an amount of space that is polynomial in the crossing number of $L$.  Thus, we may refine our triangulation in a neighborhood of each boundary component so that each of the boundary slopes is identified as a curve in the 1-skeleton.  Now we may efficiently perform an appropriate triangulated Dehn filling of each torus using our favorite method.  {\em Crushing} and {\em layered solid tori} are two popular and efficient techniques, but completely elementary ways are also efficient.  ({\em E.g.}, we may glue a triangulated disk to each slope, where the triangulation is simply the cone of the slope's triangulation; we may then fill in the resulting 2-sphere boundaries with triangulated 3-balls by taking cones.)

Now it is an exercise in definition chasing---in particular, falling back to Kirby and Melvin's earlier paper~\cite{KirbyMelvin:RT} and comparing to~\cite{TambaraY98}---to see that their category of interest is $\calB \simeq \ty(A,\chi,\nu)$, where
$A = \Ztwo$, $\chi(a,b) = \exp(\pi i ab)$ and $\nu = -1/\sqrt{2}$ as in the statement of the theorem we are currently proving.  Thus, we have that the Tambara-Yamagami state sum invariant of $M_f$ using $\ty(A,\chi,\nu)$ satisfies
\[
\begin{aligned}
|M_f|_{\ty(A,\chi,\nu)} &= \tau_{\ty(A,\chi,\nu)}(M_f)\overline{\tau_{\ty(A,\chi,\nu)}(M_f)} =\tau_{\ty(A,\chi,\nu)}(M_f)^2\\
& = \left(2\times \#\{(x_1,\dots,x_n) \in \Ztwo^n \mid f(x_1,\dots,x_n) = 0 \mod 2\} - 2^n\right)^2.
\end{aligned}\]
By arguing along the same lines of~\cite{BurtonMS18}, we may conclude that the problem of computing $|M_f|_{\ty(A,\chi,\nu)}$ from a \emph{triangulation} of $M_f$ is \#P-hard.
\qed

\end{document}